\documentclass[a4paper, 10pt, conference]{ieeeconf}    
\IEEEoverridecommandlockouts                              
\usepackage{makeidx}  

\usepackage[latin1]{inputenc}

\usepackage{amssymb}
 \usepackage{amsmath}
\include{epsf}
\usepackage {epsfig}
\usepackage{graphicx}
\usepackage{setspace}
\newtheorem{theorem}{Theorem}[section]
\newtheorem{lemma}[theorem]{Lemma}

\newtheorem{corollary}[theorem]{Corollary}

\newtheorem{proposition}[theorem]{Proposition}

\title{\LARGE \bf
Fuzzy Consensus and Synchronization: Theory and Application to Critical Infrastructure Protection Problems.
}
\author{Stefano Panzieri$^*$, Gabriele Oliva$^*$ and Roberto Setola$^\dagger$\\
$^*$ Dipartimento di Informatica e Automazione, University ``Roma TRE'',\\
Via della Vasca Navale, 79, 00146, Roma, Italy.\\
panzieri@uniroma3.it, oliva@dia.uniroma3.it\\
$^\dagger$ University Campus Biomedico of Rome, Italy\\
r.setola@unicampus.it
}

\begin{document}

\maketitle
\thispagestyle{empty}
\pagestyle{empty}

\begin{abstract}
In this paper the Distributed Consensus and Synchronization problems  with fuzzy-valued initial conditions are introduced, in order to obtain a shared estimation of the state of a system based on partial and distributed observations, in the case where such a state is affected by ambiguity and/or vagueness.
The {\em Discrete-Time Fuzzy Systems} (DFS) are introduced as an extension of scalar fuzzy difference equations and some conditions for their stability and representation are provided. 

The proposed framework is then applied in the field of Critical Infrastructures; the consensus framework is used to represent a scenario  where human operators, each able to observe directly the state of a given infrastructure (or of a given area considering vast and geographically dispersed infrastructures), reach an agreement on the overall situation, whose severity is expressed in a linguistic, fuzzy way; conversely synchronization is used to provide a distributed interdependency estimation system, where an array of interdependency models is synchronized via partial observation.

\end{abstract}


\section{INTRODUCTION}
\label{sec_introduction}

The mathematical modeling of real systems is generally subject to non trivial issues; in fact, in many cases, the phenomenon to be modeled is excessively complex, and the resulting model may either be a simplistic model or a too complex representation, which, in both cases, risks to be unsuitable for practical use.
Indeterminacy and vagueness often arise as a consequence of our inability to exactly distinguish events in real situations; classical methodologies fail to cope with such vagueness.
 
Notice that, if the nature of errors is random or probabilistic, then it is possible to adopt a stochastic framework; however if the underlying structure is not probabilistic, for instance due to subjective modeling choices, a different formalism is required.

The role of indeterminacy and vagueness is particularly relevant when the system being modeled directly involves humans; for instance in the case where the state of the system represents a subjective preference or belief.

In order to address such a problem and provide an adequate modeling tool, the fuzzy formalism appears the most natural choice.

In the literature the \emph{synchronization} of identical distributed systems has been widely investigated \cite{Synchro2,Synchro4,Tuna1,Tuna2}. \emph{Synchronization} is intended as the convergence of the solutions of an array of identical systems to a common trajectory; the synchronization approach is said to be \emph{distributed} if each system receives data only by a subset of the other systems (i.e. only from its \emph{neighborhood}). When the trajectory is a stationary point (or a double integrator), the problem reduces to a \emph{consensus} problem \cite{Olfati1,Olfati2}. However a framework able to handle the synchronization of systems in the presence of uncertain and vague values has not yet been introduced.
Indeed \emph{uncertainty} and \emph{vagueness} is often present in many applicative contexts and, in particular, is necessarily introduced when human experts are involved.
This is especially true in the field of \emph{Critical Infrastructures} (CI) interdependency modeling, where models are often tuned by means of the information provided by human stakeholders and operators \cite{Lewis,SetolaMarino,Kelly:2001} or where, during crisis scenarios, human operators are in charge to estimate the effects of outages in widely dispersed and geographically distributed infrastructures.

In order to address the distributed consensus and synchronization problems under uncertainty, in this paper the \emph{Discrete-Time Fuzzy Systems} (DFS) are introduced and their stability is studied; moreover it is proven that, under rather general hypotheses, if the initial conditions of a crisp (i.e., non-fuzzy) system are fuzzyfied, then its stability properties are preserved.
The distributed consensus and synchronization problem are then extended in the fuzzy fashion, considering fuzzy initial conditions.

The paper is organized as follows: 
some preliminary definitions are collected in Section \ref{Preliminaries}, while the Discrete-time Fuzzy Systems are introduced in Section \ref{fuzzysystems};
Sections \ref{consensus}  reviews the distributed consensus problem and its fuzzy extension, while Section and \ref{synchronization} reviews  the distributed synchronization problem and its fuzzy extension; the application of such methodologies in the field of Critical Infrastructure Protection is discussed in Section \ref{application}, providing also some simulative case studies; finally some conclusive remarks are discussed in Section \ref{conclusions}, while, in order to ease the reading, some proofs are collected in the Appendix.

\section{Preliminaries}
\label{Preliminaries}
In the following, 
${ x}$ will denote a vector with fuzzy entries, while {\em crisp} (i.e., non-fuzzy) vectors will be denoted by ${ z}$.

Let $I_p$ denote the $p\times p$ identity matrix and let ${ c}_p$ be a vector of $p$ components, each equal to $c$.

Let $A \otimes B$ denote the Kronecker product of two square matrices $A$ ($n\times n$) and $B$ ($m\times m$), that is the $nm \times nm$ matrix:
\begin{equation}
\label{kronecker}
A \otimes B=\begin{bmatrix}
a_{11}B&\cdots&a_{1n}B\\
\vdots&\ddots&\vdots\\
a_{n1}B&\cdots&a_{nn}B\\
\end{bmatrix}
\end{equation}

Let $\mathbb{R},\mathbb{N}$ be the set of reals and integers, respectively and $\mathbb{R}^+,\mathbb{N}^+$ be the set of nonnegative real and integer numbers, respectively.
Let $\mathbb{K}_C^n$ be the space of nonempty compact convex subsets of $\mathbb{R}^n$.
Let $\mathbb{B}^n$ be the open unit ball in $\mathbb{R}^n$.

Given a space $\mathbb{X}$ and a particular distance $d_{\mathbb{X}}$ defined over such a space, the Hausdorff separation and the Hausdorff metric for two sets $A,B \subset \mathbb{X}$ are given by: 
\begin{eqnarray}
\rho^*_{d,\mathbb{X}}(A,B)=sup\{d_{\mathbb{X}}(a,b) : a\in A, b \in B\} \\
\rho_{d,\mathbb{X}}(A,B)=\max \{\rho^*_{d,\mathbb{X}}(A,B), \rho^*_{d,\mathbb{X}}(B,A) \}
\end{eqnarray}


Consider the following discrete-time crisp (i.e., non-fuzzy) system:
\begin{equation}
\label{Ieq:crispdiffsist0}
\begin{matrix}
{ z}(k+1)={ G}({ z}(k),k), & { z}(0)= { z}_0
\end{matrix}
\end{equation}
where $k\in \mathbb{N}^+$ represents the discrete time step, ${ G}: \mathbb{R}^n \times \mathbb{N}^+ \rightarrow \mathbb{R}^n$ is continuous and ${ z},{ z}_0 \in \mathbb{R}^n$.

Let ${ 0}_n=[0,\ldots,0]^T\in \mathbb{R}^n$; ${ 0}_n$ is a stable solution for System (\ref{Ieq:crispdiffsist0}) if, for each $\epsilon >0$  there exists a positive function $\delta(\epsilon)$ such that
\begin{equation}
\label{dds}
d_{\mathbb{R}^n}({ z}_0, { 0}_n)< \delta(\epsilon) \mbox{ implies } d_{\mathbb{R}^n}({ z}(k), { 0}_n) < \epsilon, \forall k\geq 0
\end{equation}

Moreover if further than (\ref{dds}) it is also verified that 
\begin{equation}
\lim_{k\rightarrow +\infty}d_{\mathbb{R}^n}({ z}(k), { 0})= 0
\end{equation}
 then ${ 0}_n$ is said to be a \emph{asymptotically stable} solution of (\ref{Ieq:crispdiffsist0}).

If a matrix M is non-negative (non-positive), i.e., it has only non-negative (non-positive) entries, write $M\geq 0$ ($M\leq 0$).
If $B-A$ is a nonnegative matrix, write $B\geq A$.

Let $g_i({ z}(k),k):\mathbb{R}^n \times \mathbb{N}^+  \rightarrow \mathbb{R}$ be the $i-th$ component of ${ G}({ z}(k),k)$ (i.e., ${ G}({ z}(k),k)$ is the column-wise concatenation of the $g_i({ z}(k),k)$, for all $i=1,\ldots n$).

Let ${ z}_a, { z}_b \in \mathbb{R}^n$; define the partial orderings $\geq$ as follows
\begin{equation}
\label{pord}
{ z}_a \geq { z}_b \Leftrightarrow z_{ai} \geq z_{bi}, \forall i=1,\ldots, n
\end{equation}

The definition of $\leq$ is analogous.

A continuous function $g_i({ z},k): \mathbb{R}^n \times \mathbb{N}^+ \rightarrow \mathbb{R}$ is said to be \emph{monotone nondecreasing} in ${ z}$ if:
\begin{equation}
{ z}_a \geq { z}_b \Rightarrow g_i({ z}_a,k) \geq g_i({ z}_b,k)\quad \forall k\in \mathbb{N}^+,\quad \forall { z}_a,{ z}_b\in \mathbb{R}^n
\end{equation}

Conversely, it is \emph{monotone nonincreasing} if:
\begin{equation}{ z}_a \leq { z}_b \Rightarrow g_i({ z}_a,k) \leq g_i({ z}_b,k)
\end{equation}
or in other terms if $g_i(\cdot,k)$ preserves the partial ordering $\leq$ (or $\geq$).

The above definition can be easily  extended to ${ G}({ z},k)$ if it preserves the partial ordering $\leq$ ($\geq$) defined in (\ref{pord}).

Notice that the monotone nondecreasing assumption for system \eqref{Ieq:crispdiffsist0} is equivalent to require that the system is {\em non-negative}, that is, its state ${ z}(k)$ remains nonnegative for each $k\in \mathbb{N}^+$ if the initial conditions ${ z}_0$ are nonnegative.
%
%
%

Let $\mathcal{G}=\{\mathcal{V},\mathcal{E},\Gamma\}$ be a {\em weighted graph} with $p$ nodes, where  the set $\mathcal{V}$ denotes the nodes ${\it v}_1, \ldots, {\it v}_p$; $\mathcal{E}$ is the set of edges $({\it v}_i,{\it v}_j)$.
Matrix ${ \Gamma}=\{\gamma_{ij}\}$ is the {\em weighted adjacency matrix} describing the network topology; it is composed of non-negative entries and  $\gamma_{ij} > 0 \Leftrightarrow ({\it v}_i,{\it v}_j) \in \mathcal{E} $,
i.e., there exists an arc that starts form node $v_i$ and reaches node $v_j$. The value of $\gamma_{ij}$ represents the weight of the arc.

The graph $\mathcal{G}$ is said to be \emph{undirected} if $({\it v}_j,{\it v}_i)\in \mathcal{E}$ whenever $({\it v}_i,{\it v}_j)\in \mathcal{E}$ (the weights are $\gamma_{ij}=\gamma_{ji}$ in this case); otherwise the graph $\mathcal{G}$ is said to be \emph{directed}.

An undirected  graph $\mathcal{G}$ contains a {\em spanning tree} if there is a tree composed of a subset of the edges that connects each node $v_i$ to each node $v_j$; in this case the undirected graph is said to be {\em connected}.
A  directed graph $\mathcal{G}$ contains a {\em directed spanning tree} for a node $v_i$ if there is a tree composed of a subset of the (oriented) edges that connects $v_i$ to each other node $v_j$; in this case the directed graph is said to be {\em simply connected}.

A directed graph $\mathcal{G}$ is {\em strongly connected} if for each couple of nodes ${\it v}_i$ and ${\it v}_j$ there is a path composed of edges in $\mathcal{E}$ that connects ${\it v}_i$ and ${\it v}_j$, respecting the orientation of the edges.
Clearly a connected undirected graph is also strongly connected.

The set of the {\em neighbors} of a node ${\it v}_i$ is denoted by $\mathcal{N}_i=\{{\it v}_j\in\mathcal{V} : ({\it v}_i,{\it v}_j)\in \mathcal{E}\} $. 

Let $\mathrm{deg}^{out}_{i}=\sum_{j=1}^p \gamma_{ij}$ be the {\it out}-degree of node ${\it v}_i$ (i.e., the number of outgoing links, weighted by the $\gamma$ coefficients) and let $\mathrm{deg}^{in}_{i}=\sum_{j=1}^p \gamma_{ji}$ be the {\it in}-degree of node ${\it v}_i$ (i.e., the number of incoming links, again, weighted by the $\gamma$ coefficients). Clearly, if the graph $\mathcal{G}$ is undirected the in-degree and the out-degree coincide for each node.
Generally speaking, if the in-degree equals the out-degree, the graph $\mathcal{G}$ is said to be {\em balanced}.

Let ${L}$ be the {\em Laplacian matrix} induced by the adjacency matrix $\Gamma$ of a graph $\mathcal{G}$, whose elements $\{l_{ij}\}$ are in the form:
\begin{equation}
\label{eq:Laplacian}
l_{ij}=
\begin{cases}
\sum_{k=1}^p \gamma_{ik}, & j=i\\
-\gamma_{ij}, & j\neq i
\end{cases}
\end{equation}

From this definition, it follows that $L=D-\Gamma$, where $D$ is a diagonal matrix containing the out-degrees of each node.
Note that, by construction, $L$ has zero row-sum, i.e., $\sum_{j=1}^p l_{ij}=0$ for all $i=1,\ldots, p$.
It is a well known result \cite{Olfati1,Wu:2005} that the eigenvalues of L are all non-negative, $\lambda=0$ being its smallest eigenvalue, whose multiplicity is equal to the number of connected components of the corresponding graph.
Moreover if the graph is strongly connected, $\lambda=0$ is a simple eigenvalue of $L$ with eigenvector $1_p$ (i.e., $L 1_p$ = 0). Moreover, if the graph is also balanced (or undirected), then $1_p^TL=0$ (i.e., $L$ has also zero column-sum).

\section{Discrete-Time Fuzzy Systems}
\label{fuzzysystems}

\begin{figure}[!ht]
\begin{center}
\includegraphics[width=3.5in]{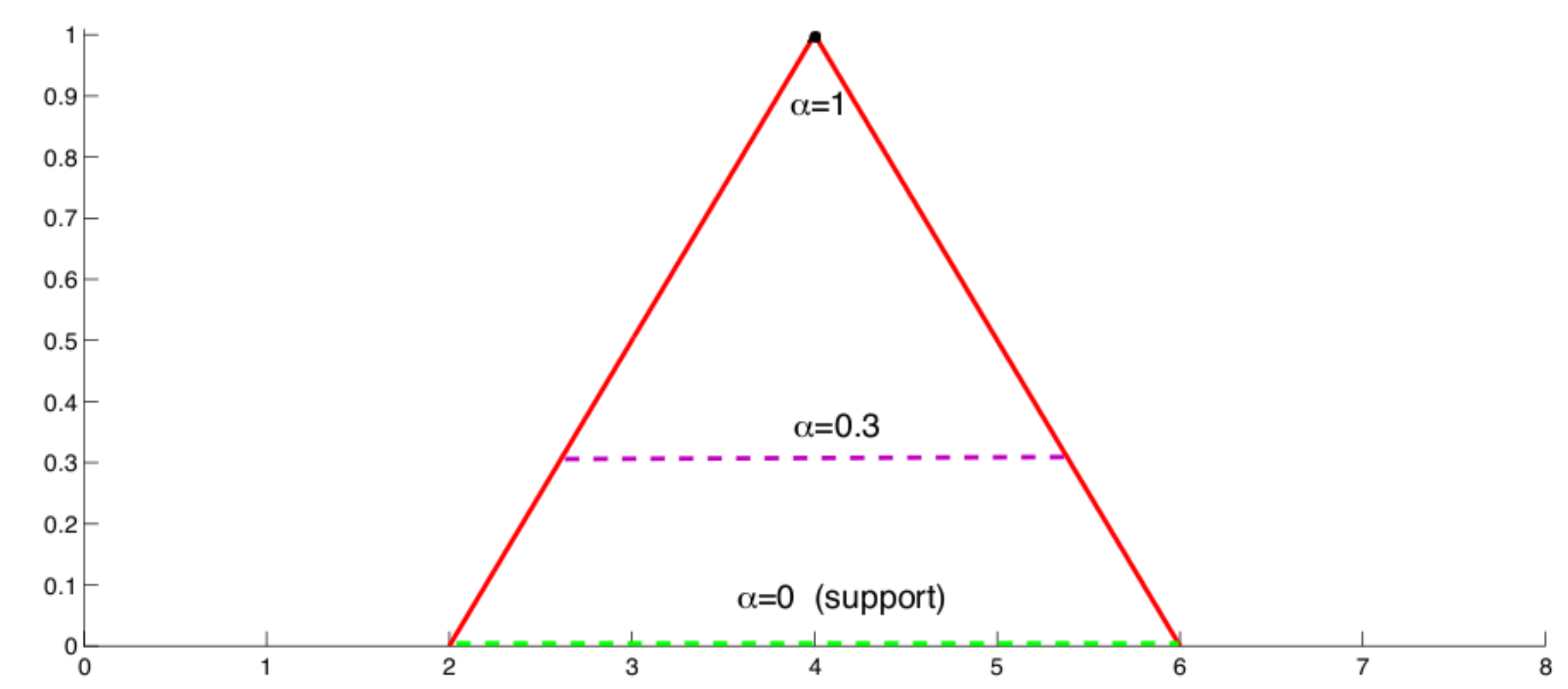}
\caption{Different $\alpha$-levels of a triangular-shaped fuzzy membership function. In particular the support ($\alpha=0$) coincides with the base of the triangle (i.e., the interval of real numbers on the abscyssae between the two endpoints), while for $\alpha=1$, due to the particular shape considered, a single (i.e., crisp) value is obtained. }
 \label{fig:alfa}
\end{center}
\end{figure}

\begin{figure}[!ht]
\begin{center}
\includegraphics[width=3.5in]{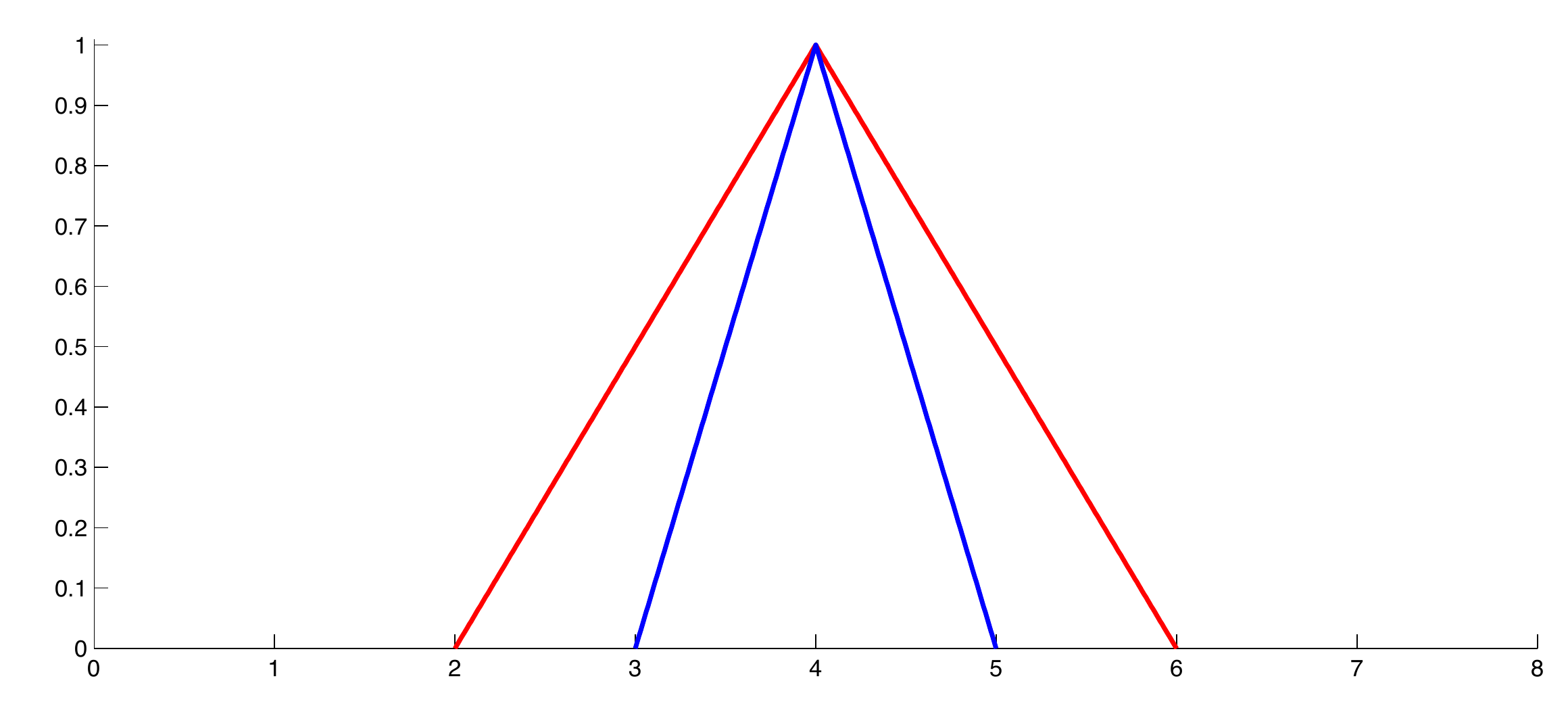}
\caption{ The two triangular fuzzy numbers depicted represent the same value ``about 4"; however the smaller one in blue is characterized by less uncertainty.}
 \label{fig:incertezza}
\end{center}
\end{figure}

\begin{figure}[!ht]
\begin{center}
\includegraphics[width=3.5in]{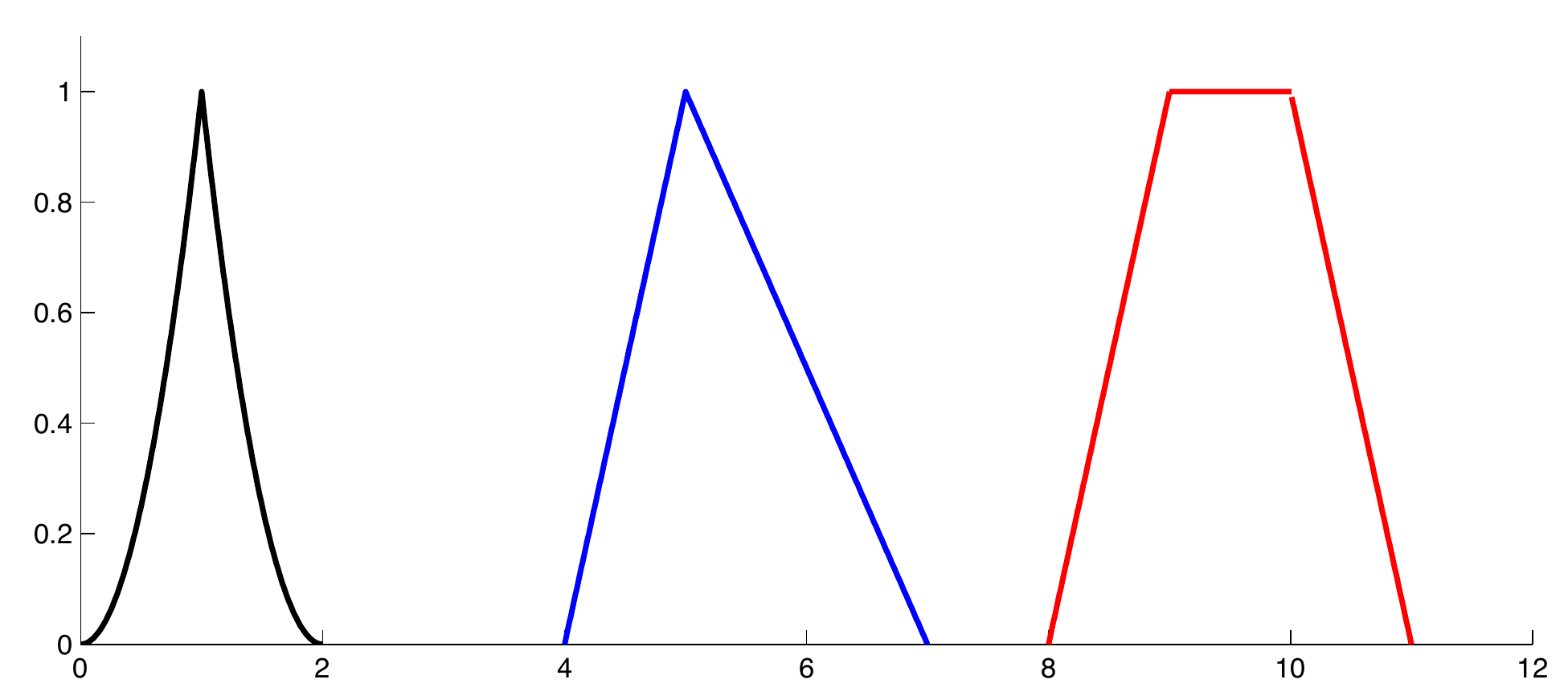}
\caption{Examples of fuzzy numbers $\mu \in \mathbb{E}$; many shapes are possible, although the Triangular Fuzzy Numbers (in blue) are the most used, because they can be described by the triple of the abscissae of their vertices ($\{4,5,7\}$ in this case). More complex shapes, however, allow to characterize better the uncertainty; the leftmost Fuzzy Number (in black) represents a case where uncertainty rapidly decreases while approaching the central value;  the rightmost Fuzzy Number, due to its trapezoidal shape, models the case where a single value with maximum belief can not be found.
Notice further that the shape of a FN needs not to be symmetric, thus allowing to represents different beliefs on the left and right spread of uncertainty with respect to the value associated with the maximum belief.    }
 \label{fig:shape}
\end{center}
\end{figure}

Let a fuzzy subset of $\mathbb{R}$ be defined in terms of a membership function $\mu:\mathbb{R} \rightarrow (0,1]$ which assigns to each point $p \in \mathbb{R}$ a grade of membership in the fuzzy set; such function is used to denote the corresponding fuzzy set.
Let the $p$-membership $\mu(p)$ be defined as the grade of membership of $p\in \mathbb{R}$ in the set $\mu$.

For each $\alpha\in[0,1]$, the $\alpha$-level set $[\mu]^\alpha$ of a fuzzy set is the subset of points $p \in\mathbb{R}$ with membership grade $\mu(p)\geq \alpha$.
The support $[\mu]^0$ of a fuzzy set is defined as the closure of the union of all its $\alpha$-level sets.
 
Let $\mathbb{E}$ be the space of all fuzzy subsets $\mu$ of $\mathbb{R}$ such that:
\begin{enumerate}
\item $\mu$ maps $\mathbb{R}$ onto $[0,1]$;
\item $[\mu]^0$ is a bounded subset of $\mathbb{R}$;
\item $[\mu]^\alpha$ is a compact subset of $\mathbb{R}$ for all $\alpha \in (0,1]$;
\item $\mu$ is \emph{fuzzy convex}, that is: $\mu(\phi p + (1-\phi)q) \geq min [\mu(p),\mu(q)]$ for all $p,q \in \mathbb{R}$
\end{enumerate}
the fuzzy sets of $\mathbb{E}$ are often called \emph{Fuzzy Numbers} (FN).
Indeed, as shown in figure (\ref{fig:alfa}), the use of $\alpha$-levels allows to address fuzzy numbers as a set of real intervals, as will be explained in Section \ref{levelwise}.

A \emph{triangular} fuzzy number (TFN) $\mu\in \mathbb{E}$, in particular, is described by an ordered triple $\{\mu_l,\mu_c,\mu_r\} \in\mathbb{R}^3$ with $\mu_l\leq \mu_c\leq \mu_r$ and such that $[\mu]^0=[\mu_l,\mu_r]$ and $[\mu]^1=\{\mu_c\}$, while in general the $\alpha$-level set is given, for any $\alpha \in [0,1]$ by:
\begin{equation}
[\mu]^\alpha = [\mu_c - (1-\alpha)(\mu_c -\mu_l), \mu_c + (1-\alpha)(\mu_r-\mu_c)]
\end{equation}

Figure (\ref{fig:incertezza}) shows two TFNs that represent the fuzzyfication of the same \emph{crisp} (i.e., non-fuzzy) number 4; notice that the width of the base of the triangle is a measure of the uncertainty associated to the TFN. Triangular representation is not the sole available alternative; as depicted in Figure (\ref{fig:shape}) many other shapes are possible, and the more complex is the shape, the more descriptive is the resulting fuzzy number (i.e., the uncertainty is  better characterized). For instance the existence of a plateau for a given interval represents complete indeterminacy for that interval, or, in the case of
 risk impact analysis, an asymmetry with respect to the peak may represents different beliefs for the best and worst cases.

The space  $\mathbb{E}$ is typically \cite{Laksh:2003} equipped with the following metric
\begin{eqnarray}
\label{hausdist}
d_\mathbb{E}(\mu,\nu)=\sup_{\alpha >0} \{\rho_\mathbb{R}([\mu]^\alpha,[\nu]^\alpha)\} & \mu,\nu \in \mathbb{E}
\end{eqnarray}
The {\em level-wise convergence} (i.e., the convergence of the $\alpha$-levels of the system) is defined as follows.
Let $\{\mu_n\}$ be a sequence on $\mathbb{E}$, then $\{\mu_n\}$ converges level-wise to $\mu \in \mathbb{E}$ if, for all $\alpha \in (0,1]$:
\begin{eqnarray}
\rho_{\mathbb{R}}([\mu_n]^\alpha,[\mu]^\alpha)\rightarrow 0 & \mbox{ as } & n\rightarrow \infty
\end{eqnarray}
Define  
\begin{equation}
\label{eq:psieq}
\Psi =  \{ \mu \in \mathbb{E}: \mu(\phi p + (1-\phi) q) \geq \phi \mu(p) + (1-\phi)\mu(q)\}
\end{equation}
for each $p, q \in [\mu]^0, \phi \in [0,1]$.
In  \cite{Laksh:2003} it is proved that convergence in $(\mathbb{E},d_{\mathbb{E}})$ implies level-wise convergence; moreover limiting to sequences in $\Psi$, the implication among convergence in $(\mathbb{E},d_{\mathbb{E}})$ and level-wise convergence can be revised.
The following theorem extends the above result to TFNs.
\begin{theorem}
\label{lem:tfn}
Limiting to sequences in the set of TFNs, level-wise convergence implies convergence in $(\mathbb{E},d_{\mathbb{E}})$.
\end{theorem}
\begin{proof}
See 	Appendix.
\end{proof}\vspace{2mm}

In order to consider vectors of $N$ components, each being a FN, the space $\mathbb{E}^{N}$ has to be characterized; to this end we chose to equip $\mathbb{E}^{N}$ with the following metric:
\begin{equation}
\label{dizero}
d_{\mathbb{E}^N}(x,y) = \sum_{i=1}^N d_\mathbb{E}(x_i,y_i)
\end{equation}
where $x=[x_1, \ldots,x_n]^T$ and $y=[y_1, \ldots,y_n]^T$, $x,y \in \mathbb{E}^{N}$.

Define the $\alpha$-level of a vector of FNs $\nu\in \mathbb{E}^N$ as the set of vectors $z\in \mathbb{R}^N$ such that, $\forall i =1,\ldots, N$ $z_i $ belongs to the $\alpha$-level of $i$-th component $\nu_i$.

Let $\{\nu_n\}$ be a sequence on $\mathbb{E}^N$, then $\{\nu_n\}$ converges level-wise to $\nu \in \mathbb{E}^N$ if, for all $\alpha \in (0,1]$:
\begin{eqnarray}
\rho_{\mathbb{R}^N}([\nu_n]^\alpha,[\nu]^\alpha)\rightarrow 0 & \mbox{ as } & n\rightarrow \infty
\end{eqnarray}
Define  
\begin{equation}
\label{eq:psieq}
\begin{matrix}
\Psi^N =  \{ \nu \in \mathbb{E}^N:  \nu_i(\phi p + (1-\phi) q) \geq\\
\geq \phi \nu_i(p) + (1-\phi)\nu_i(q), \forall i=1,\ldots, N\}
\end{matrix}
\end{equation}
for each $p, q \in [\nu_i]^0, \phi \in [0,1]$.
Notice that, with the above definition of $\alpha$-level of a vector of FNs, we have that
\begin{equation}
\lim_{n\rightarrow \infty} \rho_{\mathbb{R}^N}([\nu_n]^\alpha,[\nu]^\alpha)= 0
\Leftrightarrow
\lim_{n\rightarrow \infty} \rho_{\mathbb{R}}([\nu_{ni}]^\alpha,[\nu_i]^\alpha)= 0 
\end{equation}
$\forall i=1,\ldots, N$, therefore the scalar results on level-wise convergence are extended to the vectorial case.

Define a {\em Discrete-Time Fuzzy System } (DFS) as follows: 
\begin{equation}
\label{fuzzysystem0}
\begin{matrix}
x(k+1) = F(x(k),k);& x(0)=x_0
\end{matrix}
\end{equation}
where $x,x_0\in \mathbb{E}^N$ and
$F \in \mathbb{F}_c[\mathbb{N}^+\times \mathbb{E}^{N}, \mathbb{E}^{N}]$.

In the following Section the Stability of DFS will be addressed.

\subsection{Stability of DFS}
\label{stabilityDFS}
Let a DFS system in the form 
\begin{equation}
\label{fuzzysystem}
\begin{matrix}
x(k+1) = F(x(k),k);& x(0)=x_0
\end{matrix}
\end{equation}
where $x,x_0\in \mathbb{E}^N$
and let a crisp system in the form
\begin{equation}
\label{eq:crispdiffsist} 
\begin{matrix}
z(k+1)=G(z(k),k), & z(0)= z_0
\end{matrix}
\end{equation}
where $z,z_0 \in \mathbb{R}^N$ and $G\in \mathbb{F}_c[\mathbb{N}^+\times \mathbb{R}^N,\mathbb{R}^N]$.

Analogously to the crisp case, let $\hat 0\in\mathbb{E}^{N}$ denote the trivial solution of Eq. (\ref{fuzzysystem}), which we assume to exist.
The trivial solution $\hat 0$ of  System (\ref{fuzzysystem}) is stable if, for each $\epsilon >0$  there exists a positive function $\delta(\epsilon)$ such that
\begin{equation}
d_{\mathbb{E}^N}[x_0, \hat 0]< \delta(\epsilon) \mbox{ implies } d_{\mathbb{E}^N}[x(k), \hat 0] < \epsilon, \forall k\geq 0
\end{equation}
If $d_{\mathbb{E}^N}[x(k), \hat 0]\rightarrow 0$, as $k\rightarrow +\infty$, then the trivial solution $\hat 0$ of  System (\ref{fuzzysystem}) is said to be \emph{asymptotically stable}.

Let us now discuss the solution of (\ref{fuzzysystem}) in terms of solutions
of the crisp system (\ref{eq:crispdiffsist}), extending the results reported in \cite{Trig,Laksh:2003}.

\begin{theorem}
\label{lem:stab3}
Let a DFS system (\ref{fuzzysystem}) and let a crisp (\ref{eq:crispdiffsist})
where $G(z(k),k)$ is a continuous function,  monotone-nondecreasing with respect to  $z(k)$ for $k\geq 0$. 

Suppose that exists a continuous and positive valued \emph{defuzzyfication function} 
$V(x(k),k): \mathbb{E}^{N}\times \mathbb{N}_+  \rightarrow \mathbb{R}_+^N$ such that, posing $z(0)=V(x(0),0)$, for each $k\geq 0$
\begin{equation}
\label{dpiuineq}
V(x(k+1),k+1) \leq G(V(x(k),k),k)
\end{equation}
Suppose further that there exists a continuous, monotone non-decreasing function $a(\cdot)$  defined in $\mathbb{R}\rightarrow \mathbb{R}_+$ such that
\begin{equation}
\label{bound_condition}
a(d_{\mathbb{E}^N}[x(k),\hat 0])\leq V_0(x(k),k)
\end{equation}
where $d_{\mathbb{E}^N}$ is the distance in $\mathbb{E}^N$ defined in (\ref{dizero}) and $V_0(x(k),k)$ is defined as 
\begin{equation}
V_0(x(k),k)=\sum_{j=1}^N V_i(x(k),k),\quad  \forall k\geq 0
\end{equation} 

Then the stability properties of the trivial solution of Eq.(\ref{eq:crispdiffsist}) imply the corresponding stability properties of the trivial solution of Eq.(\ref{fuzzysystem}).
\end{theorem}
\begin{proof}
See Appendix.
 \end{proof}\vspace{3mm}

The above comparison theorem is very useful, since the stability of fuzzy systems can be derived from the stability of a non-fuzzy system.
The following Corollary provides a useful parallelism between the stability of a fuzzy system and the stability of a crisp system obtained by defuzzyfication (i.e., the case in which $G(z(k),k)= F(z(k),k)$).

\begin{corollary}
\label{teolin}
Let a DFS in the form of Eq. (\ref{fuzzysystem}) and let the following crisp systems
\begin{eqnarray}
\label{xxxpippo}
z(k+1)= F(z(k),k)&z(0)=V(x(0),0)
\end{eqnarray}
such that $F(z,k)$ is monotone nondecreasing in $z$ for each $k\in \mathbb{N}^+$; then the stability properties of crisp system (\ref{xxxpippo}) imply the corresponding stability properties of  the DFS (\ref{fuzzysystem}).
\end{corollary}
\begin{proof}
See Appendix.
\end{proof}\vspace{3mm}

\subsection{Levelwise Representation}
\label{levelwise}

Each state variable $x_i$ is such that, at any time $k$ its $\alpha$-level is given by
\begin{equation}
\label{fuzzyintervals}
[x_i(k)]^\alpha = [\underline x_i^\alpha(k),\bar x_i^\alpha(k)], \quad \forall i=1, \ldots, N
\end{equation}

In \cite{Pearson,Seikkala,Kay}, it is shown that, for each time $k$ and for each $\alpha$-level , the evolution of the system can be described by $2N$ crisp difference equations for the endpoints of the intervals of eq. (\ref{fuzzyintervals}). 

In the linear and stationary case where $F(x,k)=Fx$, (i.e., $F$ is a $N\times N$ matrix), matrix $F$ is decomposed into a monotone nondecreasing $F^+$ and a monotone nonincreasing $F^-$.
The evolution for each $\alpha$-level is given by

\begin{equation}
\label{POSNEG}
\begin{bmatrix}
 \underline x^\alpha(k+1)\\
\bar x^\alpha(k+1)\\
\end{bmatrix}
\mbox{ }=
\mbox{ }
\begin{bmatrix}
F^+ & F^-\\
F^- & F^+
\end{bmatrix}
\begin{bmatrix}
 \underline x^\alpha(k)\\
\bar x^\alpha(k)\\
\end{bmatrix}
\end{equation} 

Indeed, under the monotonicity assumption it follows that $F^+=F$ and $F^-=0$; therefore the stability of System (\ref{POSNEG}) is assured by the stability of the crisp system $z(k+1)=Fz(k)$, i.e., the above stacked system is composed of two isolated replica of the original model. 

Dropping monotonicity assumption ($\geq$), however, $F^-\neq 0$.

Let the transform matrix $T=\begin{bmatrix}I&-I\\I&I\end{bmatrix}$; changing coordinates by means of this transformation we have that 
\begin{equation}
\begin{matrix}
T\begin{bmatrix}
F^+ & F^-\\
F^- & F^+
\end{bmatrix}T^{-1}=\begin{bmatrix}
F^+-F^- & 0\\
0& F^+ + F^-
\end{bmatrix}
=\begin{bmatrix}
|F| & 0\\
0& F
\end{bmatrix}
\end{matrix}
\end{equation}
Since the transformed matrix is block diagonal, the stability, for each $\alpha$-level is obtained if both $F$ and $|F|$ are stable in the discrete-time sense or, in other terms, the stability of the absolute valued dynamic crisp system is required to assure the stability of the fuzzy system.

\section{Consensus}
\label{consensus}
Let $\mathcal{G}=\{\mathcal{V},\mathcal{E},\Gamma\}$ be a graph with $p$ nodes, and assume without loss of generality that the graph has unitary weights.

Let us consider without loss of generality the scalar consensus, and let ${z}_i\in\mathbb{R}$ be the state variable associated to the $i$-th node.

Nodes $i$ and $j$ are said to {\em agree} if ${z}_i={z}_j$ and consequently the graph $\mathcal{G}$ {\em agrees} if each couple of nodes $v_i$ and $v_j$ agrees, for all $v_i,v_j\in\mathcal{V}$. Whenever all the nodes of a network are in agreement, the common value of all nodes is called the {\em group decision value}.

Let each node in the network be modeled as a discrete-time {\em dynamic agent}, whose dynamics is in the form:
\begin{equation}
\label{eq:dynamicagent}
\begin{matrix}
{z}_i(k+1)=g({z}_i(k),{e}_i(k)),& {z}_i(0)={z}_{i0}
\end{matrix}
\end{equation}
where ${e}_i(k)\in\mathbb{R}$ represents the input for $i$-th agent and ${z}_{i0}$ is the initial condition vector for $i$-th agent.
The stacked dynamics for all the $p$ agents is given by:
\begin{equation}
\label{stack_dyn_agent}
{ z}(k+1)=G({ z}(k),{ e}(k)) \quad { z}(0)=\begin{bmatrix}{z}_{1}(0), \cdots, {z}_{p}(0)\end{bmatrix}^T
\end{equation}
where ${ z}, { e} \in \mathbb{R}^{p}$ are stack vectors composed of the state variables and inputs of each agent, respectively, i.e., ${ z}=\begin{bmatrix}{z}_1&\cdots&{z}_p\end{bmatrix}^T$ and ${ e}=\begin{bmatrix}{e}_1&\cdots&{e}_p\end{bmatrix}^T$; $G({ z}(k),{ e}(k))$  is the column-wise concatenation of the elements $g({z}_i(k),{e}_i(k))$.

Let a function $\chi:\mathbb{R}^{p} \rightarrow \mathbb{R}$; the {\em $\chi$-consensus problem} can be interpreted as a distributed way to calculate $\chi({ z}(0))$ by using as inputs for each node only information depending on the values of its neighbors $\mathcal{N}_i$.

Define a {\em protocol}, i.e.,
\begin{equation}
{e}_i(k)=f_i({z}_{j_1}(k),\ldots,{z}_{jm_i}(k))
\end{equation}
with $j_1,\ldots,j_{m_i } \in   \mathcal{N}_i  \cup \{i\}$ and, obviously, $m_i< p$. A protocol asymptotically solves the $\chi$-consensus problem if and only if there exists an asymptotically stable equilibrium ${z}^*$ of (\ref{stack_dyn_agent}) such that for each node ${z}^*_i=\chi({ z}(0))$ for all $i\in[1,p]$.

In the literature different typologies of consensus have been addressed; in the following we will review the 
average consensus problem for networks of discrete-time single and double integrators.
\subsection{Single Integrators}
Consider a network composed of $p$ dynamic agents, each one described by an integrator \cite{Olfati1}:
\begin{equation}
\label{eq:dynamicagent_firstorder_cont}
\begin{matrix}
\dot z_i(t)=e_i(t),&z_i(0)=z_{i0}&\forall v_i\in\mathcal{V}
\end{matrix}
\end{equation}
where $ z_i\in \mathbb{R}$, $ e_i\in \mathbb{R}$.
In the discrete-time fashion the system becomes:
\begin{equation}
\label{eq:dynamicagent_firstorder}
\begin{matrix}
z_i(k+1)=z_i(k)+\tau e_i(k),&z_i(0)=z_{i0}&\forall v_i\in\mathcal{V}
\end{matrix}
\end{equation}
where $\tau >0$ represents the sampling time.
In \cite{Olfati1} the following protocol is used to solve the continuous-time average consensus problem:
\begin{equation}
\label{eq:protocol_firstorder}
\begin{matrix}
e_i(t)=\sum_{j\in \mathcal{N}_i}\mbox{ }a_{ij}[z_j(t)-z_i(t)]
\end{matrix}
\end{equation}
where $a_{ij}$ are the coefficients of the adjacency matrix of the considered graph.
The resulting stacked dynamic system for the $p$ agents is given by
\begin{equation}
\label{eq:overalldynamic_firstorder}
\begin{matrix}
\dot { z}(t)=-L{ z}(t), & { z}(0)={ z}_0
\end{matrix}
\end{equation}
where ${ z}\in \mathbb{R}^p$ and $L$ is the {\em graph Laplacian} induced by $\Gamma$.

In \cite{WeiRen} it is proved that, if the graph $\mathcal{G}$ is simply connected (i.e., contains at least a directed spanning tree) a consensus equal to a linear combination of the initial conditions of the agents is reached, while the consensus coincides with the actual average of the initial conditions if the graph is strongly connected and balanced.

When the protocol is applied in the discrete-time fashion, the resulting stacked dynamics is in the form:
\begin{equation}
\label{eq:overalldynamic_firstorder_discrete}
\begin{matrix}
{ z}(k+1)=P_\tau { z}(k)
\end{matrix}
\end{equation}
where $P_\tau=I_{p}-\tau L$ is called the {\em Perron matrix} \cite{Olfati1}. 

The following Lemma \cite{Olfati1} provides a stability condition for the discrete time first order average consensus problem.
\begin{lemma}
\label{lemma_fo}
Let $\tau_1^*=1/l^*$; then choosing $\tau < \tau_1^*$, protocol (\ref{eq:protocol_firstorder}) solves the consensus problem for a network of discrete-time single integrator agents if the graph $\mathcal{G}$ contains at least a directed spanning tree.
If the graph $\mathcal{G}$ is also strongly connected and balanced, then the average consensus is achieved.
\end{lemma}
\vspace{3mm}

Notice that, as illustrated in the following, the condition on the sampling rate $\tau<\tau_1^*$ is a sufficient and conservative estimation of the maximum sampling rate that guarantees the convergence of the consensus problem.

\subsection{Double Integrators}
Consider a network composed of $p$ dynamic agents, each one described by a {\emph double integrator} \cite{WeiRen0,WeiRen01,Olfati12}:
\begin{equation}
\label{eq:dynamicagent_secondorder}
\begin{matrix}
\ddot z_i(t)=e_i(t),& \dot z_i(0)=\dot z_{i0},& z_i(0)=z_{i0},& \forall v_i\in\mathcal{V}
\end{matrix}
\end{equation}
where $z_i\in \mathbb{R}$ and $e_i\in \mathbb{R}$.

In \cite{Olfati12} the protocol used to solve the continuous-time average consensus problem is:
\begin{equation}
\label{eq:protocol_firstorder_continuous}
\begin{matrix}
e_i(t)=\sum_{j\in \mathcal{N}_i}a_{ij}[\dot z_i(t)-\dot z_i(t)]+\sum_{j\in \mathcal{N}_i}a_{ij}[z_i(t)-z_j(t)]
\end{matrix}
\end{equation}
Defining ${ z}_a={ z} \in \mathbb{R}^p$ and ${ z}_b=\dot { z} \in \mathbb{R}^p$, the resulting dynamic for the $p$ agents can be posed in the form
\begin{equation}
\label{eq:dynamicagent_secondorder}
\begin{cases}
\dot { z}_{a}(t)={ z}_{b}(t)\\
\dot { z}_{b}(t)=-L{ z}_{a}(t) - { z}_{b}(t)+{ e}(t)
\end{cases}
\end{equation}
where ${ z}_a(0)={ z}_{0}=\begin{bmatrix}z_{1}^0, \cdots, z_{p}^0\end{bmatrix}^T$ and ${ z}_b(0)=\dot { z}_{0}=\begin{bmatrix}\dot z_{1}^0, \cdots, \dot z_{p}^0\end{bmatrix}^T$.

In order to obtain a discrete-time representation of an agent described by a double integrator, in \cite{WeiRen1,WeiRen2} the above model is sampled with sample time $\tau$, and the following protocol is adopted
\begin{equation}
\label{eq:protocol_secondorder_weiren}
\begin{matrix}
e_i(k)=\sum_{j\in \mathcal{N}_i}a_{ij}[(z_{ai}(k)-z_{aj}(k))]+\\ 	\\
+\sum_{j\in \mathcal{N}_i}a_{ij}[z_{bi}(k)-z_{bj}(k))]
\end{matrix}
\end{equation}
providing a condition for the stability of the resulting system.

In this paper, however, we will adopt a simpler formulation, considering two inputs $e_{ai}$ and $e_{bi}$. Hence the model of the $i$-th agent, in the proposed formulation, becomes:
\begin{equation}
\label{eq:dynamicagent_secondorder}
\begin{cases}
\dot z_{ai}(t)=z_{bi}(t)+e_{ai}(t)\\
\dot z_{bi}(t)=e_{bi}(t)
\end{cases}
\end{equation}
and consequently its discretization with sample time $\tau$ is
\begin{equation}
\label{eq:dynamicagent_secondorder}
\begin{cases}
z_{ai}(k+1)=z_{ai}(k)+\tau z_{bi}(k)+\tau e_{ai}(k)\\
z_{bi}(k+1)=z_{bi}(k)+\tau e_{bi}(k)
\end{cases}
\end{equation}

\begin{theorem}
\label{teo_so}
Let $p$ systems in the form of Eq. (\ref{eq:dynamicagent_secondorder}), such that their graph $\mathcal{G}$ contains at least a directed spanning tree. 
Then if $\tau<\frac{1}{l^*}$ the agents reach a consensus using the protocol
\begin{equation}
\label{eq:protocold_secondorder}
\begin{cases}
e_{ai}(k)=\sum_{j\in \mathcal{N}_i}a_{ij}[(z_{ai}-z_{aj})]\\
e_{bi}(k)=\sum_{j\in \mathcal{N}_i}a_{ij}[(z_{bi}-z_{bj})]
\end{cases}
\end{equation}

Moreover if the graph $\mathcal{G}$ is also strongly connected and balanced, then the consensus reaches the average of the initial conditions.
\end{theorem}
\begin{proof}
The overall dynamic for the $p$ systems, considering the protocol (\ref{eq:protocold_secondorder}) becomes
\begin{equation}
\label{eq:protocol_secondorder_weiren}
\begin{bmatrix}
{ z}_{a}(k+1)\\{ z}_{b}(k+1)
\end{bmatrix}
=
\begin{bmatrix}
I-\tau L & \tau I \\
0 & I-\tau L
\end{bmatrix}
\begin{bmatrix}
{ z}_{a}(k)\\{ z}_{b}(k)
\end{bmatrix}
\end{equation}
Note that, in this formulation, the dynamics of ${ z}_b$ is decoupled from ${ z}_a$, and has the standard form of a single integrator system. Therefore an agreement on ${ z}_b$ is reached if $\tau< 1/l^*$.
Since the dynamic matrix is block triangular, the condition $\tau< 1/l^*$ is sufficient for the stability and the consensus is reached asymptotically. And from Lemma \eqref{lemma_fo} if $\mathcal{G}$ contains at least a directed spanning tree a consensus is reached, while if the graph is strongly connected and balanced, the average consensus is reached.

\end{proof}
\vspace{3mm}

Note that, in this case the average consensus is an agreement on the ``velocity" which becomes constant, while the ``position", although being asymptotically the same for all the agents, varies with constant velocity.

Note that also in this case the condition on the maximum value of the sampling rate has to be considered as a conservative estimation.

\subsection{Fuzzy Consensus}
\label{consensus_fuzzy}
In this Section a framework for the discrete-time first order and second order average consensus problem with fuzzy initial condition is provided extending the results for crisp systems \cite{Olfati1}.
Let a  graph $\mathcal{G}$ composed of $p$ agents, 
where the dynamics of each agent is described by the DFS
\begin{equation}
\label{fuzzy_consensus_agent}
{ x}_i(k+1)=F({ x}_i(k),{ e}_i(k)),\quad { x}_i(0)={ x}_{i0}
\end{equation} 
Where ${ x}_i, { x}_{i0},{ e}_i \in \mathbb{E}^{q}$ ${ e}_i$ depends only on the state of agent $i$ and his neighbors $N_i$, according to the topology described by the matrix $\Gamma$.  The array of fuzzy agents reaches consensus if
\begin{eqnarray}
\label{inen2}
\lim_{k\rightarrow \infty} d_{\mathbb{E}^q}[ { x}_i(k), { x}_j(k)] =  0,  & \forall i,j=1,\ldots,p
\end{eqnarray}
where $d_{\mathbb{E}^q}$ is the metric defined in Eq. (\ref{dizero}).

Conversely, the array of fuzzy agents reaches consensus {\em level-wise} if, for all $\alpha \in [0,1]$ and $\forall i=1,\ldots,p$:
\begin{eqnarray}
\label{lwisecon}
\lim_{k\rightarrow \infty} \rho_{d_{\mathbb{R}^q},\mathbb{R}^q}({ x}_{i}^\alpha(k),{ x}_{j}^\alpha(k)) =  { 0}\end{eqnarray}

The following result correlates consensus  and level-wise consensus of fuzzy systems:
\begin{proposition}
\label{lwise}
The consensus of fuzzy agents in the sense of (\ref{inen2}) implies level-wise consensus (\ref{lwisecon}).
\end{proposition}
\begin{proof}
The proof of such a statement trivially derives from the definition of $d_{\mathbb{E}}$ \end{proof}

Note that, limiting each state variable to the set $\Psi$, defined in Eq. (\ref{eq:psieq}), the implication can be reversed; hence this is true for systems with initial conditions described by TFNs.

The following theorem correlates the crisp and fuzzy consensus in the case of single and double integrators. To avoid confusion the following notation for the fuzzy extension of single and double discrete-time integrators will be adopted, considering the state arranged by type: $y_i$ (or $[y_i , v_i ]^T$) will denote the state of the {\it i-th} single (or double) integrator; ${ y}$ (or $[{ y}, { v} ]^T$) will denote the stack array composed of the state variables of all the $p$ agents of the whole system.
 The consensus value reached by $i$-th agent (which, obviously, assumes the same value for each agent) is denoted by $y_i^*$ (or $[y_i^* , v_i ^*]^T$). We will also consider the vector of the consensus states for the array of agents, denoted by ${ y}^*$ (or $[{ y}^* , { v}^* ]^T$ ). Note that in the case of the double integrators the agents reach a constant value for ${ v}^*$ (e.g., constant velocity), while ${ y}^*$, although shared by the agents, changes (e.g., position); in this case, then, we will denote the consensus reached at {\it k-th} step as $[{ y}^*(k), { v}^*]$.
\begin{theorem}
\label{syncfz}
Consider a fuzzy consensus problem with $p$ interconnected discrete time agents with fuzzy dynamic described by Eq. (\ref{fuzzy_consensus_agent}), where the graph $\mathcal{G}$ contains at least a directed spanning tree.
Then both the  following statements hold true:
\begin{enumerate}
\item if the dynamic of each agent is a discrete-time single integrator in the form of Eq. (\ref{eq:dynamicagent_firstorder}), protocol (\ref{eq:protocol_firstorder}) solves the problem for $\tau < 1/ l^*$ and the agents reach an agreement;
\item if the dynamic of each agent is a discrete-time double integrator in the form of Eq. (\ref{eq:dynamicagent_secondorder}), protocol (\ref{eq:protocold_secondorder}) solves the problem for $\tau < \frac{1}{ l^*}$ and the agents reach an agreement.
\end{enumerate}
Moreover if the graph $\mathcal{G}$ is also strongly connected and balanced, the average consensus is achieved.
\end{theorem}
\begin{proof}
From Lemma \ref{lemma_fo} system (\ref{eq:dynamicagent_firstorder})  reaches consensus for $\tau<1/l^*$; moreover  the overall dynamic matrix has only nonnegative entries and is monotone nondecreasing. Hence by Theorem \ref{teolin} it follows that system (\ref{eq:dynamicagent_firstorder})  reaches consensus also in the fuzzy fashion. Analogously, by Theorem \ref{teo_so}, system (\ref{eq:protocol_secondorder_weiren}) reaches consensus in the fuzzy fashion for $\tau < \frac{1}{l^*}$. The same considerations apply for average consensus.
\end{proof}

\begin{corollary}
\label{pippocorollary}
Let $r$ be a vector in $\mathbb{R}^p$, with $\sum_{j=1}^p r_j =1$ and $r_j\geq 0$ for each $ j=1,\ldots,p$.
Under the hypotheses of Theorem (\ref{syncfz}) assuming that $\Gamma$ contains at least a directed spanning tree, then the following statements hold true for each $\alpha$-level:
\begin{enumerate}
\item in the case of discrete-time single integrators  
\begin{equation}
\label{pippo}
\begin{cases}
\underline{y_i}^{*\alpha} = \sum_{j=1}^p r_j \underline{y}_{i0}^{\alpha}\\
\bar{y_i}^{*\alpha} = \sum_{j=1}^p r_j \underline{y}_{i0}^{\alpha}
\end{cases} 
\end{equation}

\item in the case of discrete-time double integrators  
\begin{equation}
\begin{cases}
\underline{v_i}^{*\alpha} = \sum_{j=1}^p r_j \underline{v}_{j0}^{\alpha}\\
\bar{v_i}^{*\alpha} = \sum_{j=1}^p r_j \bar{v}_{j0}^{\alpha}\\
\underline{y_i}^{*\alpha}(k) =  \sum_{j=1}^p r_j \underline{x}_{j0}^{\alpha}+ k\tau \underline{v}^{*\alpha}\\
\bar{y_i}^{*\alpha} (k)= \sum_{j=1}^p r_j \bar{x}_{j0}^{\alpha}+k\tau \bar{v}^{*\alpha}
\end{cases} 
\end{equation}

for each $i=1,\ldots,p$
\end{enumerate}
\end{corollary}
\begin{proof}
\noindent Let $R=[{ r},\ldots,{ r}]^T$ be a $p\times p$ matrix. From Theorem \ref{syncfz} it follows that the array of single integrators reach consensus. The value reached is given by ${ x}^*=R{ x}_0$. Since R has only non-negative entires, for each $\alpha$-level the consensus reached assumes the structure of Eq. (\ref{POSNEG}), where $R^+=R$ and $R^-=0$, proving statement 1.
The proof for ${ v}^*$ in the case of double integrator is analogous, since the evolution of ${ v}$ coincides with the evolution of ${ y}$ in the case of single integrator (i.e., ${ v}$ does not depend on ${ y}$).
For what concerns ${ y}^*$ we have that:
\begin{equation}
{ y}(k+1)=P_\tau { y}(k)+\tau { v}(k)={ y}^a(k+1)+{ y}^b(k+1)
\end{equation}
Since the system is linear, the effects of the two terms ${ y}^a$ and ${ y}^b$ can be evaluated separately:
\begin{eqnarray}
{ y}^{a*}(k)=R { y}_0\\
{ y}^{b*}(k)=k \tau R { y}_0
\end{eqnarray}
Analogously to previous statement, using the level-wise representation of Eq. (\ref{POSNEG}), the statement is proved.
\end{proof}
\begin{corollary}
For $\alpha=1$ fuzzy consensus coincides with crisp consensus for both single and double integrators.
\end{corollary}
\begin{proof}
For $\alpha=1$ it follows that $\underline y^1(0)=\bar y^1(0)$ (i.e., the interval collapses into a single point).
Substituting inside system (\ref{pippo}), it follows that $\underline y^1(k)=\bar y^1(k)$, for all $k\geq 0$.
The proof is analogous in the case of double integrators.
\end{proof}

\section{Synchronization}
\label{synchronization}

In the literature the \emph{synchronization} of identical distributed systems has been widely investigated \cite{Synchro2,Wu:2005,Synchro4,Tuna1,Tuna2}. \emph{Synchronization} is intended as the convergence of the solutions of an array of identical systems to a common trajectory; the synchronization approach is said to be \emph{distributed} if each system receives data only from a subset of the other systems (i.e., only from its \emph{neighborhood}). When the trajectory is a stationary point, the problem reduces to a \emph{consensus} problem \cite{Olfati1,Olfati2}. However a framework able to handle the synchronization of systems in the presence of uncertain and vague values has not yet been introduced.
Indeed \emph{uncertainty} and \emph{vagueness}  are often present in many applicative contexts, in particular when human experts are involved.
This is especially true in the field of \emph{Critical Infrastructures} (CI) interdependency modeling, where models are often tuned by means of the information provided by human stakeholders and operators \cite{Kelly:2001,Lewis,SetolaMarino} or where, during crisis scenarios, human operators are in charge to estimate the effects of outages in widely dispersed and geographically distributed infrastructures.

Let $p$ identical linear systems, where the $i$-th system is in the form:
\begin{equation}
\label{systems}
\begin{cases}z_i(k+1)=Az_i(k)+u_i(k)\\ \\
y(k)=Cz(k)\end{cases}
\end{equation}
with $A\in \mathbb{R}^{n\times n}$, $C\in \mathbb{R}^{m\times n}$ and $m\leq n$, $z_i,u_i\in \mathbb{R}^n$, and consider the case in which the only information available for each system is given by 
\begin{equation}
\label{input1}
\begin{matrix}e_i(k)= \sum_{j=1}^p \gamma_{ij} (y_j(k)-y_i(k));\end{matrix}
\end{equation}
where the coefficients $\gamma_{ij}$ allow $i$-th system to communicate only with its neighbors, according to the $p \times p$ adjacency matrix $\Gamma$.
Set $u_i(k)=\Omega e_i(k)$; in other terms
\begin{equation}
\label{input}
\begin{matrix}u_i(k)=\Omega C \sum_{j=1}^p \gamma_{ij} (z_j(k)-z_i(k));\end{matrix}
\end{equation}

where $\Omega$ is an $n\times m$ matrix.

The above systems in the form of Eq. (\ref{systems}) are said to {\it synchronize} if
\begin{equation}
\begin{matrix}
\lim_{k \to +\infty} ||z_i(k)-z_j(k)||=0 \mbox{  ; } & \forall i,j=1,\ldots,p; & i\neq j
\end{matrix}
\end{equation}
i.e., if all the $n$ components of the different $p$ systems assume homologous values.

Let us now provide a single equation that summarizes the dynamics of the considered $p$ systems \eqref{systems}. Note that, using the input \eqref{input}, the dynamics of the single system is given by:
\begin{equation}
\label{equcomp}
z_i(k+1)=Az_i(k)+\Omega C \sum_{j=1}^p \gamma_{ij} (z_j(k)-z_i(k))
\end{equation}
Note further that in the above equation the state $z_i$ of the $i$-th system depends on the states $z_j$ of the other systems, according to the topology defined by $\Gamma$.
Let ${ z}(k)=[z_1(k)^T,\ldots, z_p(k)^T]^T$ be a stack vector containing the states of all the $p$ systems; the overall dynamics will be in the form ${ z}(k+1)=\hat A { z}(k)$, where $\hat A$ is a $np \times np$ matrix summarizing the dynamics of the $p$ agents, i.e.:
\begin{equation}
\label{stack}
\begin{matrix}
{ z}(k+1)=\hat A{ z}(k)=[I\otimes A - L \otimes \Omega C]{ z}(k)
\end{matrix}
\end{equation}
in fact the terms $Az_i(k)$ in Eq. \eqref{equcomp} are dependent only on $z_i(k)$ and therefore the stacked matrix contains a term 
$$I\otimes A=\begin{bmatrix}A&&\\ &\ddots&\\ &&A\end{bmatrix}$$

For the second part of matrix $\hat A$, note that in the summation of Eq. \eqref{equcomp}, the term $-z_i(k)$ appears $\sum_{j=1}^p \gamma_{ij}=l_{ii}$ times, while exactly $l_{ii}$ neighbors $z_j(k)$ are considered; therefore the term $\Omega C$ has to be pre-multiplied (in the kronecker sense), by $-L$, where $L$ is the Laplacian matrix.

 It is immediate to recognize that the $p$ systems (\ref{systems}) synchronize if the stacked system (\ref{stack}) is stable.
 Hence the synchronization of $p$ identical systems can be granted by choosing a ``control" matrix $\Omega$ that stabilizes the closed loop system (\ref{stack}).
 Note that the stability of the single systems does not guarantee by itself the synchronization of the trajectories of the different systems. 
 
In \cite{Wu:2005} a Lyapunov-based approach is adopted to grant the stability and then the synchronization, but it requires some hypotheses on all the eigenvalues of matrix L.
On the other hand in \cite{Tuna2}, under the hypothesis of a stable matrix $A$ and observable pair $(A,C)$, a sophisticated algorithm is used for the choice of matrix $\Omega$.
The following theorem proves that, under some additional hypotheses, such as the  positivity of the systems, the complexity of choosing $\Omega$ can be greatly reduced. 

\begin{theorem}
\label{teosynchro}
Suppose that matrix $A$ is such that: $\forall i,j=1,\ldots,n$ $a_{ij}\geq0$ and for all $i=1,\ldots,n$
\begin{equation}
\label{aijgeq0}
\sum_{j=1}^n a_{ij}\leq 1
\end{equation}
and suppose that $C^TC$ is non-singular. Let $K=\Omega C$.
If it is possible to find $\Omega$ such that $K$ is diagonal and
\begin{equation}
\label{kcond}
0\leq k_{ii}  \leq \frac{a_{ii}}{l_m}, \quad \forall i=1, \ldots n
\end{equation}
where $l_m$ is the minimum diagonal element of the Laplacian matrix induced by the adjacency matrix $\Gamma$; then System (\ref{stack}) is stable and the array of systems (\ref{systems}) synchronize.
\end{theorem}
\begin{proof}
The dynamic matrix of System (\ref{stack}) has the following structure:
\begin{equation}
\begin{small}
\begin{bmatrix}
A-l_{11}K &l_{12}K&\cdots&l_{1p}K\\
l_{21}K& \ddots&\ddots &\vdots\\
\vdots &\ddots &\ddots &l_{p-1,p}K\\
l_{p1}K&\cdots &l_{p,p-1}K &A-l_{pp}K \\
\end{bmatrix}
\end{small}
\end{equation}
From Gershgorin Circle Theorem \cite{CIRCLE} the eigenvalues of an $n \times n$ matrix M lie, in the complex plane, in the union of circles centered in $\chi(i)=m_{ii}$ with radius equal to $\rho(i)=\sum_{j=1; j\neq i}^n |m_{ij}|$.
Since $C^TC$ is non-singular, it is possible to set $\Omega=KC^\dag$; in this case, focusing on the $i$-th row of the $q$-th block row; the center is given by
\begin{equation}
\chi(q,i)=a_{ii}-\sum_{g=1}^p \gamma_{qg}k_{ii}=a_{ii}-l_{qq}k_{ii}
\end{equation}
By Condition (\ref{kcond}), it follows that $0 \leq \chi(q,i) \leq 1$.
The off diagonal elements for row $i$ of block row $q$ are all positive, hence the $|\cdot|$ is not required for the radius, which is given by:
\begin{equation}
\begin{matrix}
\rho(q,i)=\sum_{j=1,j\neq i}^n a_{ij}+ l_{qq}k_{ii}
\end{matrix}
\end{equation}
Since $\chi(q,i) \geq 0$, the system is stable in the discrete-time sense if $\chi(q,i)+\rho(q,i)\leq 1$ for each $q=1,\ldots p$ and $i=1,\ldots, n$, or in other terms if $\sum_{j=1}^n a_{ij} \leq 1$, which is verified from the hypotheses. 
\end{proof}	

Such a positivity assumption is in many cases justified; for instance, as will be shown in the case study, the class of interdependency models for critical infrastructures, where the percentage of malfunctioning (or inoperability) is intrinsically positive if no recovery action is considered.

Notice that, if $C^TC$ is non-singular, in order to obtain the desired diagonal matrix $K$ it is sufficient to set $\Omega= KC^\dag $, where $C^\dag=(C^TC)^{-1}C^T$ is the left pseudo-inverse of $C$.

In \cite{Tuna2} it is proved that, for a connected topology, the $p$ systems converge to the following shared state value:
\begin{equation}
\label{resultt}
{ z}^* (k)= (A^k\otimes r^T)\begin{bmatrix}z_{10}\\ \vdots \\ z_{p0}\end{bmatrix}
\end{equation}
\noindent where $r\in \mathbb{R}^p$ is a vector such that $r^TL=0$ and $\sum_{h=1}^p r_h=1$.

It is possible to specify some conditions, in order to further characterize the synchronization reached.
\begin{corollary}
\label{corzero}
If the graph is strongly connected and balanced, then the $p$ systems synchronize to the {\it average} evolution.
\end{corollary}

\paragraph{{ Proof}} Since systems synchronize, it follows that
\begin{equation}\label{ass_r}\begin{matrix}r^TL=0 &\mbox{ and } r^T1_p=1\end{matrix}\end{equation}
Moreover, a balanced graph ensures that $1_p^TL=0$.
Therefore the only $r$ that satisfies (\ref{ass_r}) is such that $r_j=\frac{1}{p}$ for each $j=1,\ldots,p$, proving the statement.
\vspace{3mm}

Note that it is also possible to obtain the synchronization to a weighted average or to the sum of the evolutions. It is sufficient to use the synchronization algorithm with modified initial conditions $\hat z_{i0}$ obtained from the real initial condition detected $z_{i0}$.
For example if $\hat z_{i0}= p z_{i0}$ for each system  $i$, the sum of the evolutions is obtained. Note that, to achieve this result, each system needs to know the number of systems involved in the synchronization.

\subsection{Synchronization of Fuzzy Systems}
Let us now discuss the synchronization of linear and stationary DFS systems with partial state coupling is provided, extending the results for crisp systems \cite{Tuna2}.
Let an array of $p$ DFS, each with $n$ state variables; the state of the $i$-th system is denoted by $x_i \in \mathbb{E}^{n}$, while the $j$-th component of  the $i$-th system is denoted by $x_{ij} \in \mathbb{E}$.

The array of $p$ identical DFS \emph{synchronizes} if
\begin{eqnarray}
\label{inen}
\lim_{k\rightarrow \infty} d_{\mathbb{E}^n}[ x_i(k), x_j(k)] =  0,  & \forall i,j=1,\ldots, p
\end{eqnarray}
The array of identical DFS synchronizes {\em level-wise} if, for all $\alpha \in (0,1]$:
\begin{eqnarray}
\begin{matrix}
\lim_{k\rightarrow \infty} \rho^*_{\mathbb{R}}([x_{ij}]^\alpha(k),[x_{iq}]^\alpha(k)) =  0 \\
\\
 \forall j,q =1,\ldots, n,\mbox{ and } \forall i=1,\ldots, p
\end{matrix}
\end{eqnarray}
The following proposition correlates synchronization and level-wise synchronization:
\begin{proposition}
\label{lwise}
The synchronization of an array of identical DFS implies their level-wise synchronization 
\end{proposition}
\paragraph{{ Proof}}
This fact derives from the definition of $d_{\mathbb{E}}$.

\vspace{3mm}

Note that, limiting each state variable to the set $\Psi$, defined in Eq. (\ref{eq:psieq}), the implication can be revised; according to Theorem \ref{lem:tfn}, this is true for linear systems with initial conditions described by triangular fuzzy numbers.

The following theorem correlates the synchronization of an array of linear and stationary crisp systems with the synchronization reached by the same systems in the fuzzy fashion.
\begin{theorem}
\label{syncfz}
Let the conditions required by Theorem \ref{teosynchro} hold for an array of $p$ crisp systems in the form of Eq.(\ref{systems}), where $\Omega$ is chosen according to Theorem \ref{teosynchro}; then the array of DFS, where each systems has the same dynamics of (\ref{systems}) and fuzzy initial condition, synchronize in the sense specified by (\ref{inen}). 
\end{theorem}
\paragraph{{ Proof}}
From Theorem \ref{teosynchro} it follows that the crisp overall stacked system (\ref{stack}) is stable and the array of $p$ systems synchronize.
Since $\Omega$ is such that $[A\otimes I - L\otimes K\Omega]$ has only nonnegative entries, it follows that the condition required by Theorem  \ref{teolin} is verified, and the theorem is proved.  
\vspace{1mm}

The following corollary provides a characterization of the synchronized evolution:
\begin{corollary}
For each $\alpha$-level the synchronized evolution is given by:
\begin{equation}
\label{synchroalfa}
\begin{matrix}
\underline x^{\alpha *}(k)=[r\otimes A^k] \begin{bmatrix}\underline x^\alpha_1(0)\\\vdots \\ \underline x^\alpha_p(0)\end{bmatrix};&
\bar x^{\alpha *}(k)=
[r\otimes A^k] \begin{bmatrix}\bar x^\alpha_1(0)\\\vdots \\ \bar x^\alpha_p(0)\end{bmatrix} 
\end{matrix}
\end{equation}
\end{corollary}
\paragraph{{ Proof}}
Eq. (\ref{resultt}) can be restated for the stacked systems as
 \begin{equation}
 { z}^*=(R\otimes H)  { z}_0
 \end{equation}
 where $R=1_p\otimes r$, $ { z}^*=[z^*_1,\ldots,z^*_p]^T$ (note that $z^*_i$ represents the synchronized evolution, and are all the same) and $ { z}_0=[z_1(0),\ldots,z_p(0)]^T$. 
 
 Since crisp and fuzzy systems are linear, when switching to the fuzzy variables, above equation becomes ${ x}^*=(R\otimes H ){ x}_0$.
 
 Furthermore, by Theorem \ref{lwise} we have that the array of fuzzy systems synchronize levelwise; therefore, for each $\alpha$-level, the synchronized evolution becomes:
  \begin{equation}
 \begin{bmatrix}\underline { x}^{\alpha *}(k)\\ \bar { x}^{\alpha *}(k)\end{bmatrix} 
 =
 \begin{bmatrix}
R\otimes (A^+)^k & R\otimes (A^-)^k\\
R\otimes (A^-)^k & R\otimes (A^+)^k
 \end{bmatrix}
  \begin{bmatrix}\underline { x}^{\alpha }(0)\\ \bar { x}^{\alpha}(0)\end{bmatrix} 
 \end{equation}
since the stacked dynamic matrix has nonnegative entries it follows that $A^-=0$ and the theorem is proved.  \vspace{3mm}

\begin{corollary}
Crisp synchronization coincides with fuzzy level-wise synchronization for $\alpha=1$.
\end{corollary}
\paragraph{{ Proof}}
For $\alpha=1$ it follows that $\underline x^1(0)=\bar x^1(0)$ (i.e. the interval collapses into a single point).
Substituting inside Eq. (\ref{synchroalfa}), it follows that $\underline x^1(k)=\bar x^1(k)$, for all $k\geq 0$, proving the statement.  \vspace{3mm}

Note that the maximum level of uncertainty and vagueness of a given variable $x_{ij}$ corresponds to the width ${ \xi}_{ij}$ of its support (i.e.,  $\alpha=0$); hence, for each variable, ${ \xi}_{ij}$ is a useful index, which can be adopted to measure the level of uncertainty, once the synchronization is reached. Such an evolution is given by:
\begin{equation}
\xi(k)=\bar x^0(k)-\underline x^0(k)=[r\otimes A^k] \begin{bmatrix}\underline x^\alpha_1(0) - \bar x^\alpha_1(0)\\\vdots \\ \underline x^\alpha_p(0) -\bar x^\alpha_p(0) \end{bmatrix}
\end{equation} 

\section{Consensus and Synchronization in Critical Infrastructure Protection}
\label{application}

In this section the application of fuzzy distributed consensus and synchronization to critical infrastructure protection will be discussed.
\subsection{Consensus}

Consider a scenario composed of $p$ highly interconnected infrastructures or lands, and let an operator or a team for each infrastructure or land be in charge to determine the local effects of an adverse event, such as terroristic attack, natural disaster or a distributed technological failure in the absence of a central coordination authority. 

In the following, two cases will be considered: in the first one the operators have to reach consensus on the actual severity of the failure affecting the whole scenario, on the base of their own partial observations; in the second case they have also to determine the expected evolution of the phenomena.

Each expert expresses a linguistic measurement of the perceived severity of the failure (and of its expected growth ratio) affecting his/her infrastructure or land using the expressions reported in Table 1 (and Table 2) providing, also, an estimate  about his/her confidence on the provided data, in accordance with the confidence scale of Table 3. The values are then encoded into TFNs according to the last column of Tables 1, 2 and 3. Specifically, Table 1 encodes the actual level of failure perceived by the operator, and Table 2 the expected growth/reduction rate (which can as well be negative); these numbers can be regarded as the central values of the TFNs, while the left and right endpoints are obtained by applying the confidence scale reported in Table 3.

Hence, being each operator aware only of its own domain, they need an instrument as that reported in this paper in order to reach a distributed consensus.

\begin{table}
\begin{center}
\label{tab:impact}
\scalebox{0.6}{
\begin{tabular}{|p{3cm}|p{7cm}|c|}
\hline
{\bf Perceived $\mbox{Ê}$Severity}& {\bf Description} & \small{Value} \\ \hline
\small{nothing } & \small{the event does not induce any effect on the infrastructure/land} & 0\\ \hline
\small{negligible} & \small{the event induces some very limited and geographically bounded consequences that have no direct impact on the infrastructure's or land's operativeness} &  0.025\\ \hline
\small{ very limited} & \small{the event induces some geographically bounded consequences that have no direct impact on the infrastructure's or land's operativeness} &   0.05\\ \hline
\small{limited} & \small{the event induces consequences only on subsystems/zones that have no direct impact on the infrastructure's or land's operativeness} &   0.1\\ \hline
\small{circumscribed degradation} & \small{the event induces geographically bounded consequences } &   0.2\\ \hline
\small{significant degradation} &\small{ the event significantly degrades the operativeness of the infrastructure/land}&  0.30\\ \hline
\small{severe degradation}  & \small{the impact  on the infrastructure/land is severe} & 0.500\\ \hline
\small{quite complete stop} & \small{ the impact  is quite catastrophic} & 0.700\\ \hline
\small{stop} & \small{total disruption} & 1\\
\hline
\end{tabular}
}
\end{center}
\caption{Perceived Severity estimation table.}
\end{table}

\begin{table}[htpb]
\begin{center}
\label{tab:growth}
\scalebox{0.6}{
\begin{tabular}{|p{3.3cm}|p{7cm}|c|}
\hline
{\bf Expected Growth/Reduction}&{\bf Description} & \small{\bf Value} \\ \hline
\small{steady }  &\small{the severity of the event is expected to remain constant.} & 0\\ \hline
\small{negligible} & \small{the severity of the event is expected to have a very limited growth/reduction.} &  $\pm 0.0001$\\ \hline
\small{very slow} & \small{the severity of the event is expected to grow/reduce only in the long term.} &  $\pm 0.001$\\ \hline
\small{slow} & \small{the severity of the event is expected to grow/reduce in the long term and eventually in the mid-term.} &  $\pm 0.03$\\ \hline
\small{quite slow} & \small{the severity of the event is expected to grow/reduce in the mid-term} &  $\pm 0.005$\\ \hline
\small{Not so slow} & \small{the severity of the event is expected to grow/reduce in the mid-term and eventually in the short term} &  $\pm 0.010$\\ \hline
\small{Quite Fast}  & \small{the severity of the event is expected to grow/reduce in the short-term} & $\pm 0.05$\\ \hline
\small{Fast} & \small{ the severity of the event is expected to grow/reduce significantly in the short-term} & $\pm  0.07$\\ \hline
\small{Very Fast} & \small{the severity of the event is expected to grow/reduce dramatically in the short-term} & $\pm  0.1$\\
\hline
\end{tabular}
}
\end{center}
\caption{Expected Growth estimation table.}
\end{table}

\begin{table}[htpb]
\begin{center}
\label{tab:mod}
\scalebox{0.6}{
\begin{tabular}{|p{2cm}|p{5cm}|c|c|}
\hline
{\bf Confidence}& {\bf Description} & \small{\bf Value} (severity)&\small{\bf Value} (growth)\\ \hline
{\bf * } & \small{Perfect Knowledge (no uncertainty)} & 0& 0\\ \hline
{\bf * *} & \small{Excellent confidence} &  $ \pm 0.005$&  $ \pm 0.0005$ \\ \hline
{\bf * * *} & \small{Good confidence} &  $\pm 0.050$&  $\pm 0.0050$ \\ \hline
{\bf * * * *} & \small{Relative Confidence} & $ \pm  0.100$ & $ \pm  0.0100$\\ \hline
{\bf * * * * *} & \small{Uncertain} & $ \pm  0.200 $& $ \pm  0.0200$\\ \hline
\end{tabular}
}
\end{center}
\caption{Confidence estimation  scale.}
\end{table}

Let us consider a scenario composed of 5 infrastructures and assume that their topology is a bipartite graph (see Figure \ref{fig:topo}.(a)). Such a topology may represent a scenario where some infrastructures are not able to communicate directly (e.g., due to physical ore commercial constrains)

Assuming unitary weight the corresponding laplacian $L_a$ is 
\begin{footnotesize}
\begin{eqnarray}
L_a=
\begin{bmatrix}
2 &0 &0 &-1 &-1\\
       0 &2 &0 &-1& -1\\
       0 &0 &2 &-1 &-1\\
       -1 &-1 &-1& 3 &0\\
       -1 &-1 &-1 &0 &3
\end{bmatrix}
\end{eqnarray}
\end{footnotesize}

Since $l_a^*=3$, in order to respect the condition required by Theorem (\ref{syncfz}), we have that for bipartite topology $\tau_1^*=\frac{1}{3} [s]$ and $\tau_2^*=\frac{1}{4} [s]$; however for the sake of uniformity, we chose $\tau=\frac{1}{4} [s]$,  for both single and double integrator cases. 

It is immediate to recognize that, due to the choice of $\tau$, the dynamic matrices of both single and double integrator case are composed by non-negative entries.  

Table 4 shows the initial conditions for both perceived severity and expected growth, each with the associated confidence, as well as the corresponding TFN.

Figure \ref{fig:first_startstop} shows the initial conditions and final synchronized state  in the case of single integrators.
 More specifically in a situation where only one operator observes a very bad situation (i.e., operator n. 3 sees a "quite complete stop") while all the others have no direct perception of the crisis (i.e., they estimate the event ranging from "nothing" to "circumscribed"), they distributedly agree on a circumscribed degradation crisis condition with a good confidence. Hence, the agents are able to share vague and ambiguous information in a distributed way and they reach a consensus obtaining a consistent qualification of the actual crisis. Note that the consensus is obtained both in terms of expected severity (e.g., the central endpoint of the triangle) and in terms of confidence on the estimation (i.e., the width of the base of the triangle).
 
 In the case of double integrator models , Figure \ref{fig:second_start} shows the initial conditions for the expected growth and the consensus reached.
 In this case, it is more evident that even in the presence of very different local perception of which should be the evolution of the phenomena (quite all the operators express no overlapping estimations, both in terms of magnitude and sign,
and two operators also with a strong credibility) they reach a distributed consensus or a common understanding of the effective growth of the evolution of the phenomenon, again, both in terms of magnitude and confidence. Note that the criticality in this latter framework is assumed to be growing with a constant rate, and the agents reach an agreement also on this varying quantity, both in terms of magnitude and confidence.

Notice that the reached consensus does not depend on the peculiar topology adopted. Any strongly connected and balanced topology with the proposed protocol allows to reach the same consensus. For example let us consider the ring topology of Figure (Figure \ref{fig:topo}.(b)), where each agent is able to communicate only with its nearest neighbors. In this case the Laplacian $L_b$ is 
\begin{footnotesize}
\begin{eqnarray}
L_b=
\begin{bmatrix}
1 &-1 &0 &0 &0\\
       0 &1 &-1 &0& 0\\
       0 &0 &1 &-1 &0\\
       0 &0 &0& 1 &-1\\
       -1 &0 &0 &0 &1
\end{bmatrix}
\end{eqnarray}
\end{footnotesize}

and we have that $\tau_1^*= 1[s]$ and $\tau_2^*=\frac{1}{2}[s]$. Hence also in this case $\tau=\frac{1}{4}$ satisfies the conditions of Theorem (\ref{syncfz}).

Let us initialize both the single and double integrator models with the same initial conditions used for the bipartite topology (Table IV).

Obviously the two topologies are not completely equivalent, because greater is the communication capability of the agents, faster is the consensus is reached.

This can be immediately recognized looking at the time evolution of the state variables of the different agents, as reported in Figures \ref{fig:first_alfa} and \ref{fig:first_alfa2} for the single integrator model with reference to the level-wise representation for $\alpha=0.3$ (i.e., the left and right extrema).

While in the bipartite graph topology the consensus is achieved after $9$ iterations, with the ring topology we need more than 30 iterations.

Analogously, in the case of double integrator model, Figures \ref{fig:second_alpha} and \ref{fig:second_alpha_chain} report the state variables of all the agents for the crisis severity estimation and  expected growth with reference to the level-wise representation for $\alpha=0.5$.

\begin{figure}
\begin{center}
\includegraphics[width=3in]{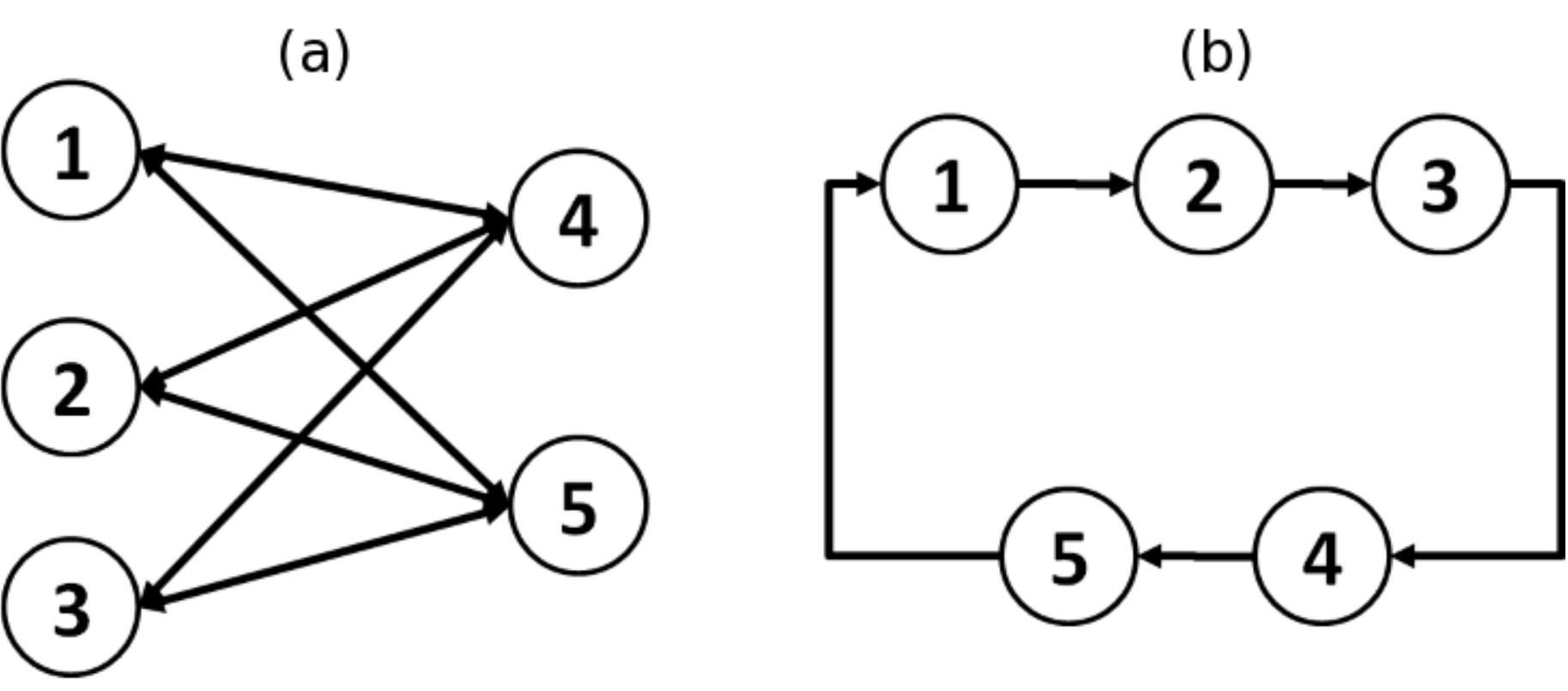}
\caption{Topologies chosen for simulations: bipartite graph (a) and chain (b). All the edges have unitary weight.}
 \label{fig:topo}
\end{center}
\end{figure}

\begin{table}[htpb]
\begin{center}
\label{tab:par}
\scalebox{0.6}{
\begin{tabular}{|c|c|l|c|c|l|c|}
\hline
n.&{\bf Severity}& {\bf Confidence} &{\bf TFN}\\ \hline
1&Nothing & {\bf * } & $[0, 0, 0]$\\ \hline
2&Limited & {\bf * * * *} & $[0, 0.1, 0.2]$\\ \hline
3&Quite Complete stop & {\bf * * } & $[0.695, 0.7, 0.705]$\\ \hline
4&Circumscribed degradation& {\bf * * } & $[0.195, 0.2, 0.205]$\\ \hline
5&Significant degradations & {\bf * * * * } & $[0.2, 0.3, 0.4]$\\ \hline
\end{tabular}
}
\end{center}

\begin{center}
\label{tab:par}
\scalebox{0.6}{
\begin{tabular}{|c|c|c|l|c|c|l|c|}
\hline
n.&{\bf Expected Growth}& {\bf Growth/Reduction}&{\bf Confidence} &{\bf TFN}  \\ \hline
1&Steady & Growth & {\bf * * * * *} &[-0.2, 0, 0.2]\\ \hline
2&Quite Fast &Growth& {\bf * *} &[0.0495, 0.05, 0.0505]\\ \hline
3& Slow &Reduction& {\bf * * *} &[-0.035, -0.03, -0.025]\\ \hline
4&Very Fast &Reduction& {\bf * * *} &[-0.105, -0.1, -0.095]\\ \hline
5&Fast &Growth& {\bf *} &[0.07, 0.07, 0.07]\\ \hline
\end{tabular}
}
\end{center}

\caption{Initial Conditions for the case study}
\end{table}
 
\begin{figure}[!ht]
\begin{center}
\includegraphics[width=3in]{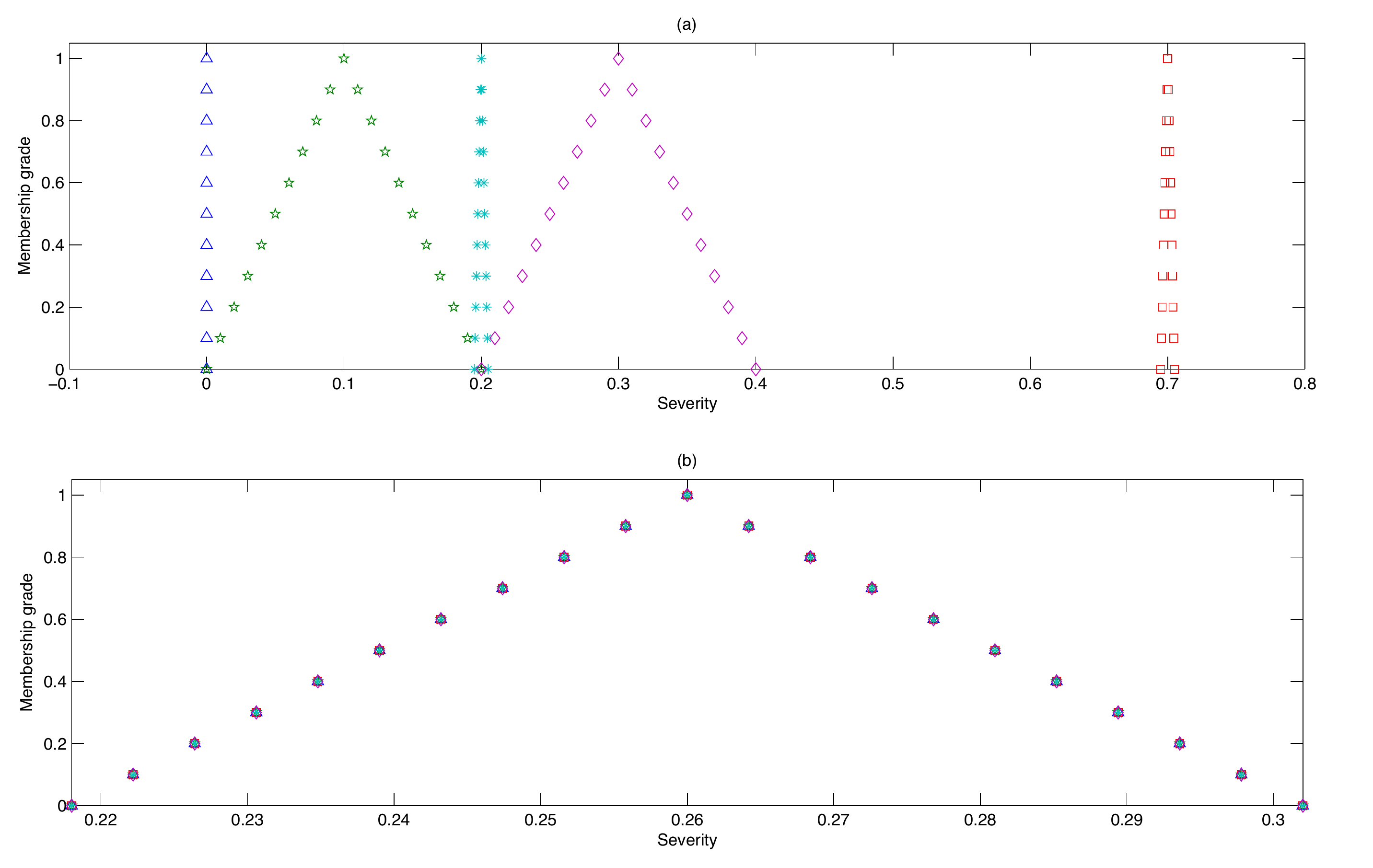}
\caption{Initial conditions (a) and synchronized state (b) for 5 discrete time single integrators. The result is the same for both topologies; however the consensus is reached after 10 steps for bipartite topology, while for chain topology 32 steps are required.}
 \label{fig:first_startstop}
\end{center}
\end{figure}

\begin{figure}[!ht]
\begin{center}
\includegraphics[width=3in]{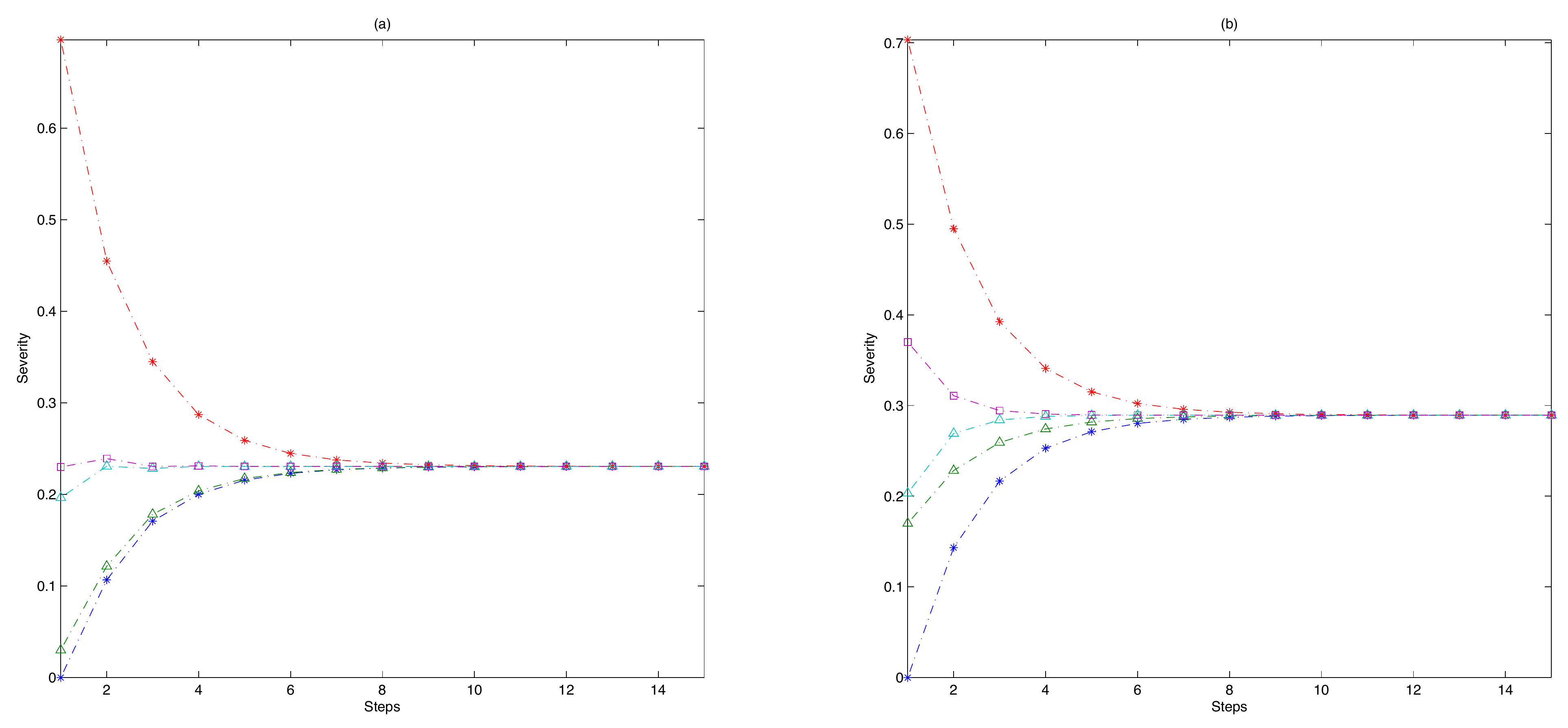}
\caption{Synchronization of left (a) and right (b) extrema of an $\alpha$-level of for 5 discrete time single integrators connected by the bipartite topology, for $\alpha=0.3$.}
 \label{fig:first_alfa}
\end{center}
\end{figure}

\begin{figure}[!ht]
\begin{center}
\includegraphics[width=3in]{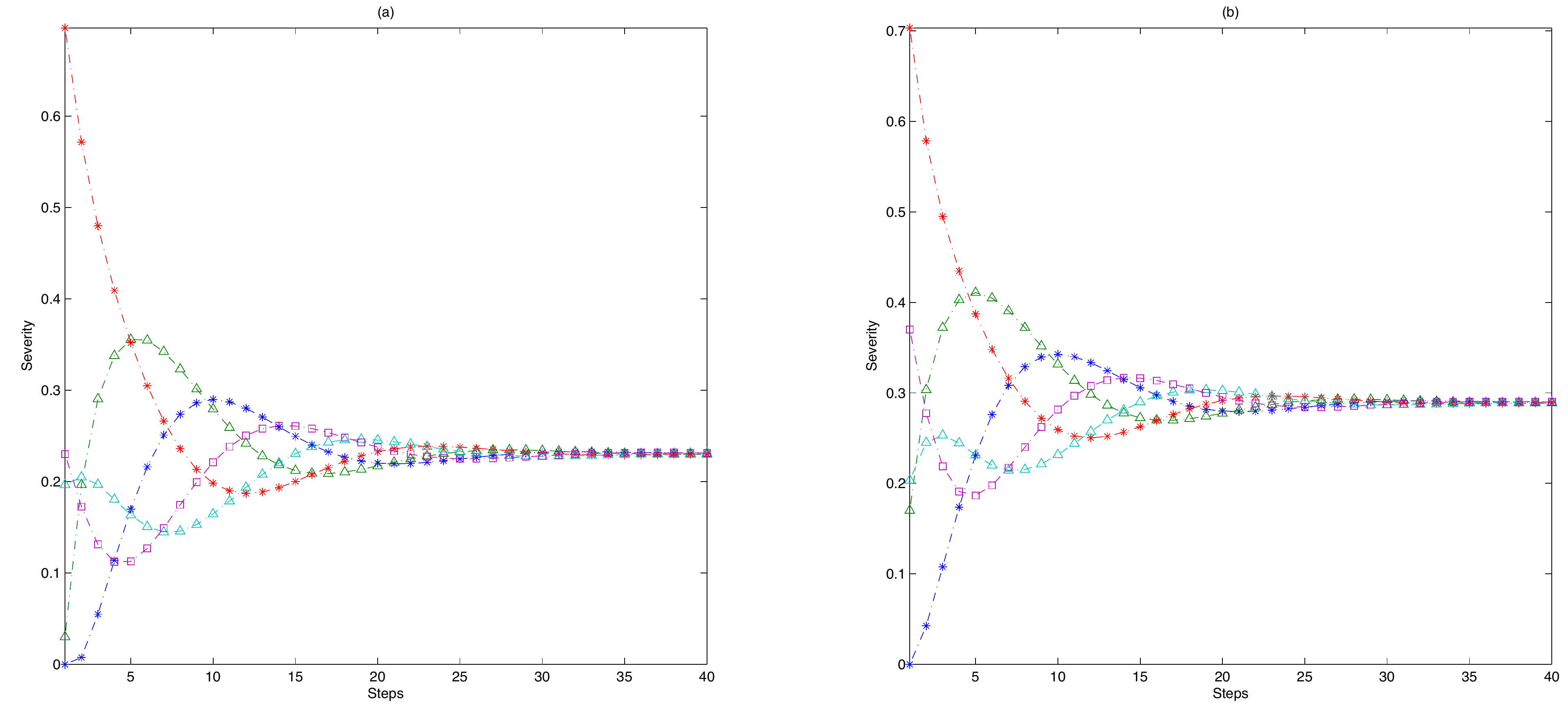}
\caption{Synchronization of left (a) and right (b) extrema of an $\alpha$-level of for 5 discrete time single integrators connected by the chain topology, for $\alpha=0.3$.}
 \label{fig:first_alfa2}
\end{center}
\end{figure}

 \begin{figure}[!ht]
\begin{center}
\includegraphics[width=3in]{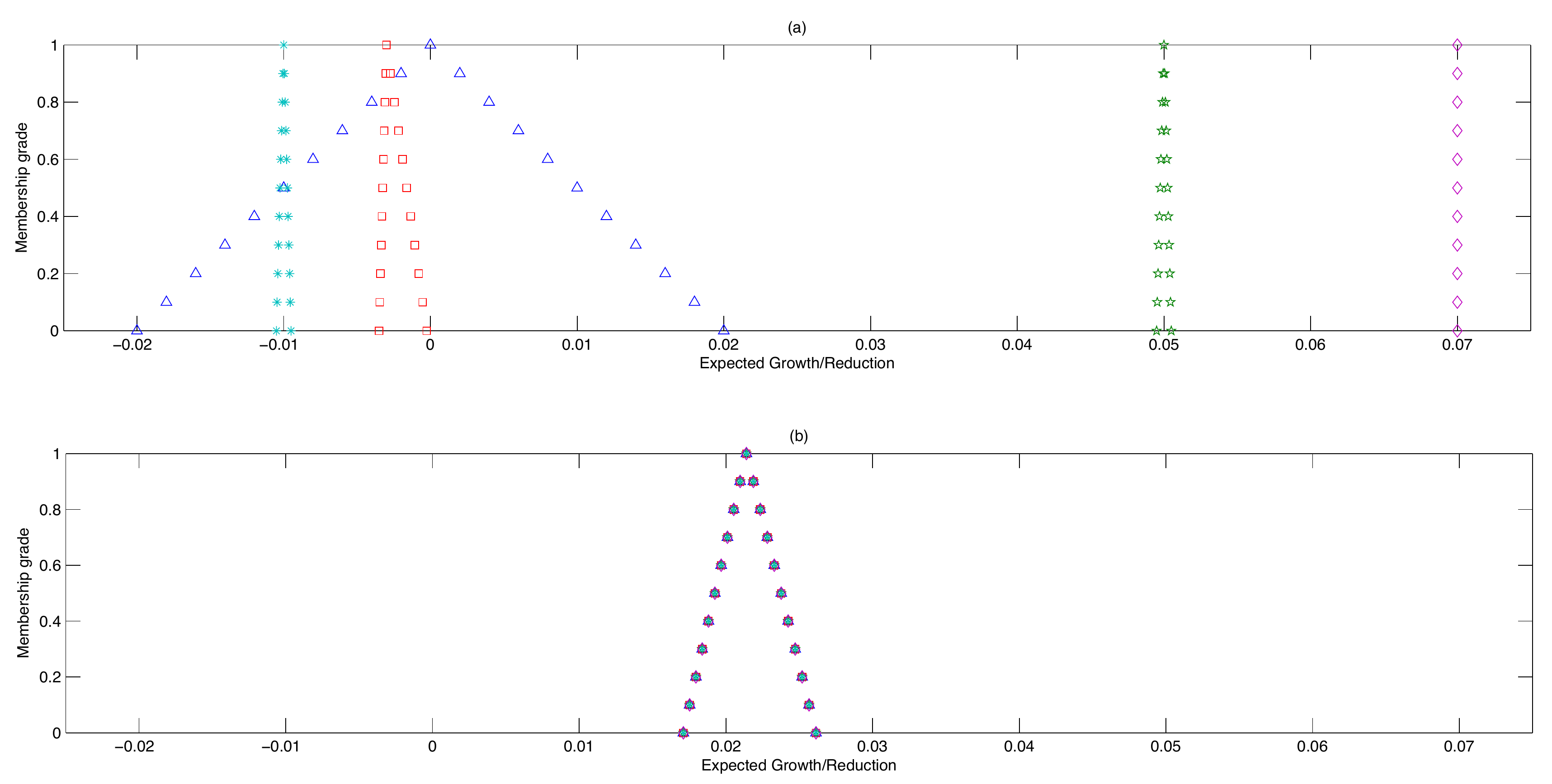}
\caption{Initial conditions for expected growth (a) and consensus reached (b) for 5 discrete time double-integrators; the consensus is reached after 8 steps for bipartite topology, and after 29 steps for chain topology.}
 \label{fig:second_start}
\end{center}
\end{figure}

\begin{figure}[!ht]
\begin{center}
\includegraphics[width=3in]{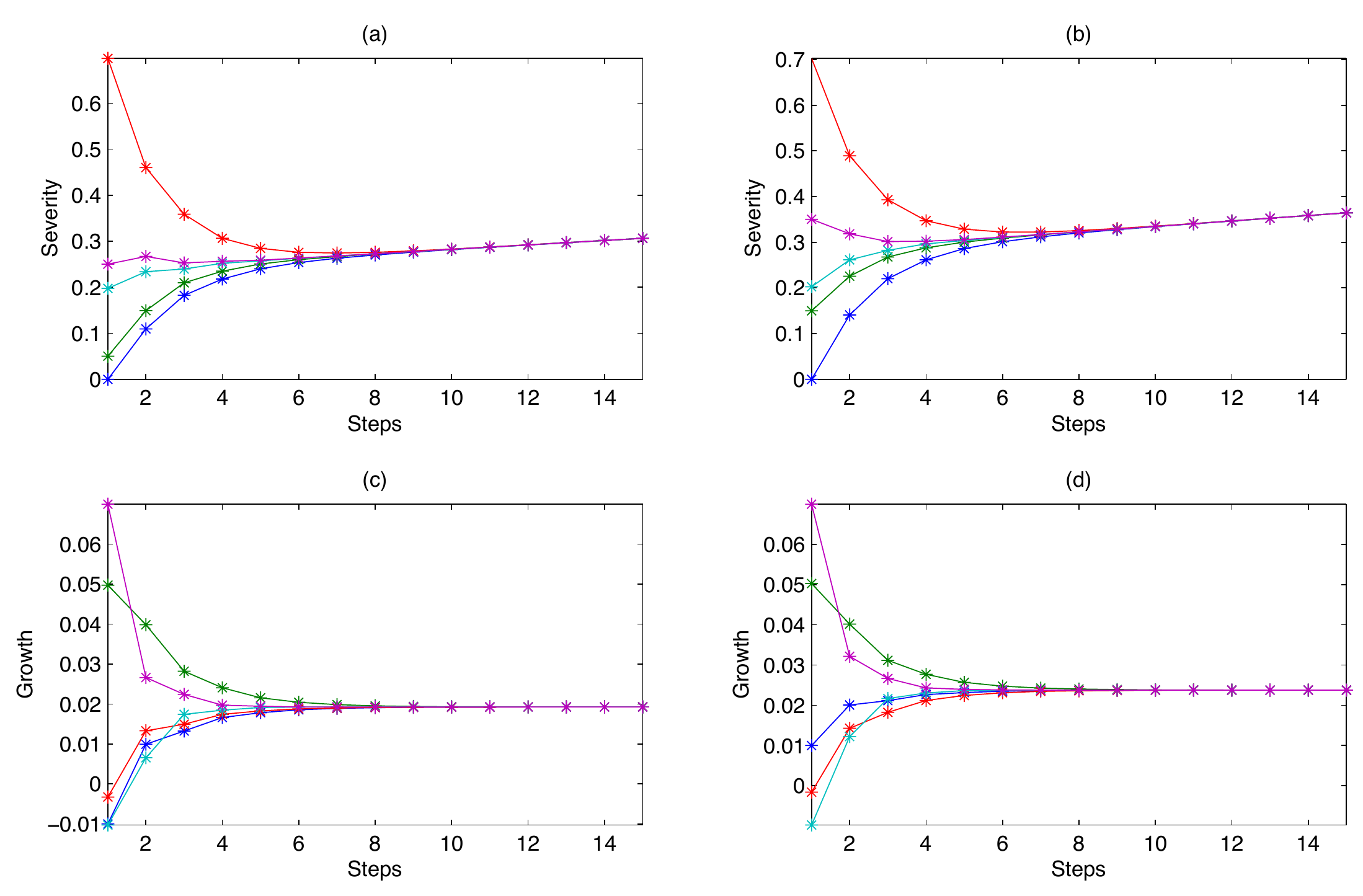}
\caption{Synchronization of an $\alpha$-level of 5 discrete time double integrators connected by the topology of Figure \ref{fig:topo}.(a) for $\alpha=0.5$: left extrema of severity (a); right extrema of severity (b); left extrema of growth (c); right extrema of growth (d).}
 \label{fig:second_alpha}
\end{center}
\end{figure}

\begin{figure}[!ht]
\begin{center}
\includegraphics[width=3in]{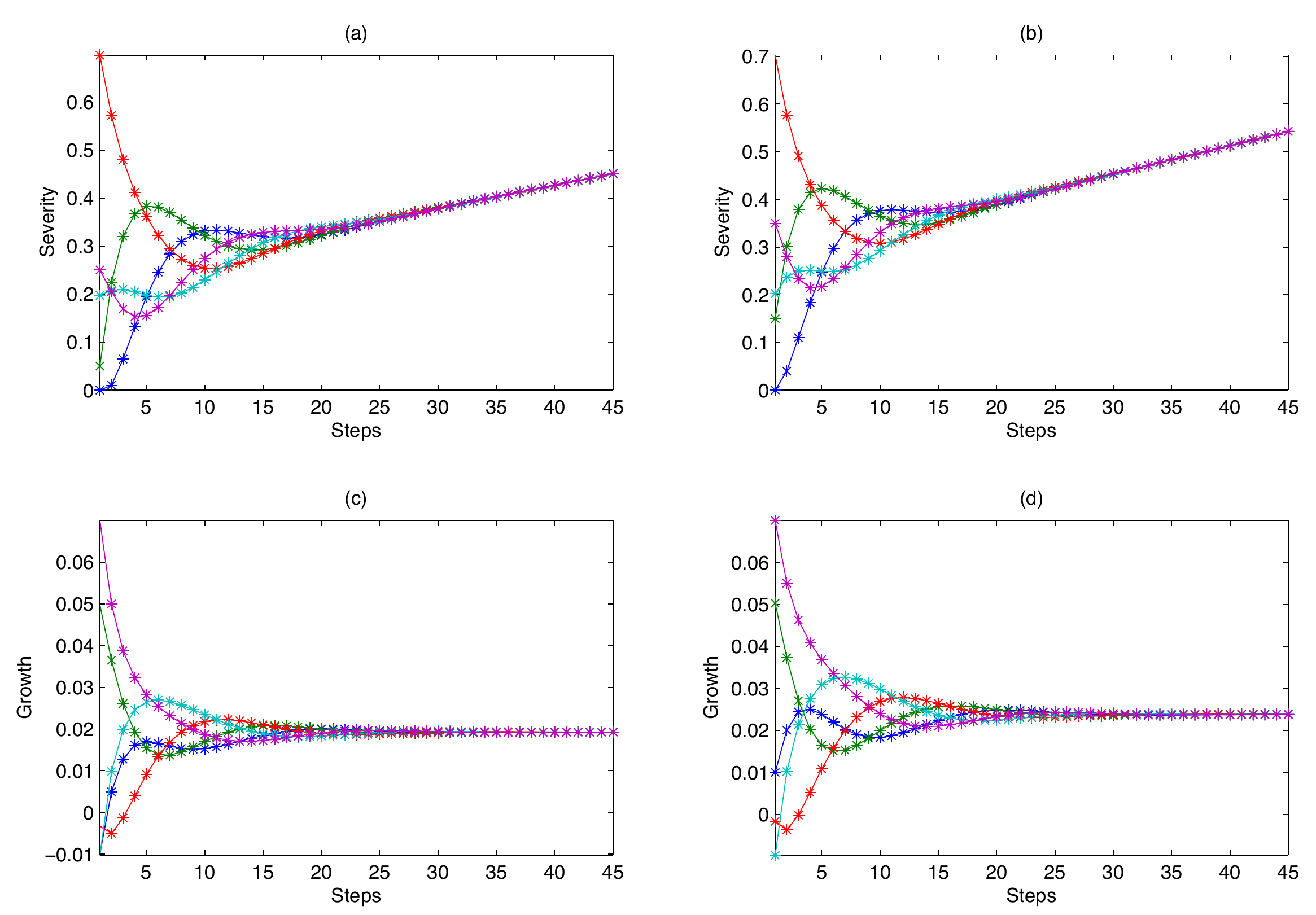}
\caption{Synchronization of an $\alpha$-level of 5 discrete time double integrators connected by the topology of Figure \ref{fig:topo}.(b) for $\alpha=0.5$: left extrema of severity (a); right extrema of severity (b); left extrema of growth (c); right extrema of growth (d).}
 \label{fig:second_alpha_chain}
\end{center}
\end{figure}

In order to better understand the influence of the sampling rate $\tau$, Figure \ref{fig:tau} shows the time required for the consensus for a given topology (i.e., $\tau k^*$, where $k^*$ is the number of steps required for consensus) for $\tau \in (0,1]$. Note that such a time is almost constant for small values of $\tau$, while it diverges for $\tau$ that reaches the stability boundaries. Finally note that the system may reaches consensus even for values of $\tau$ greater than $\tau_1^*$ or $\tau_2^*$.

 \begin{figure}[!ht]
\begin{center}
\includegraphics[width=3in]{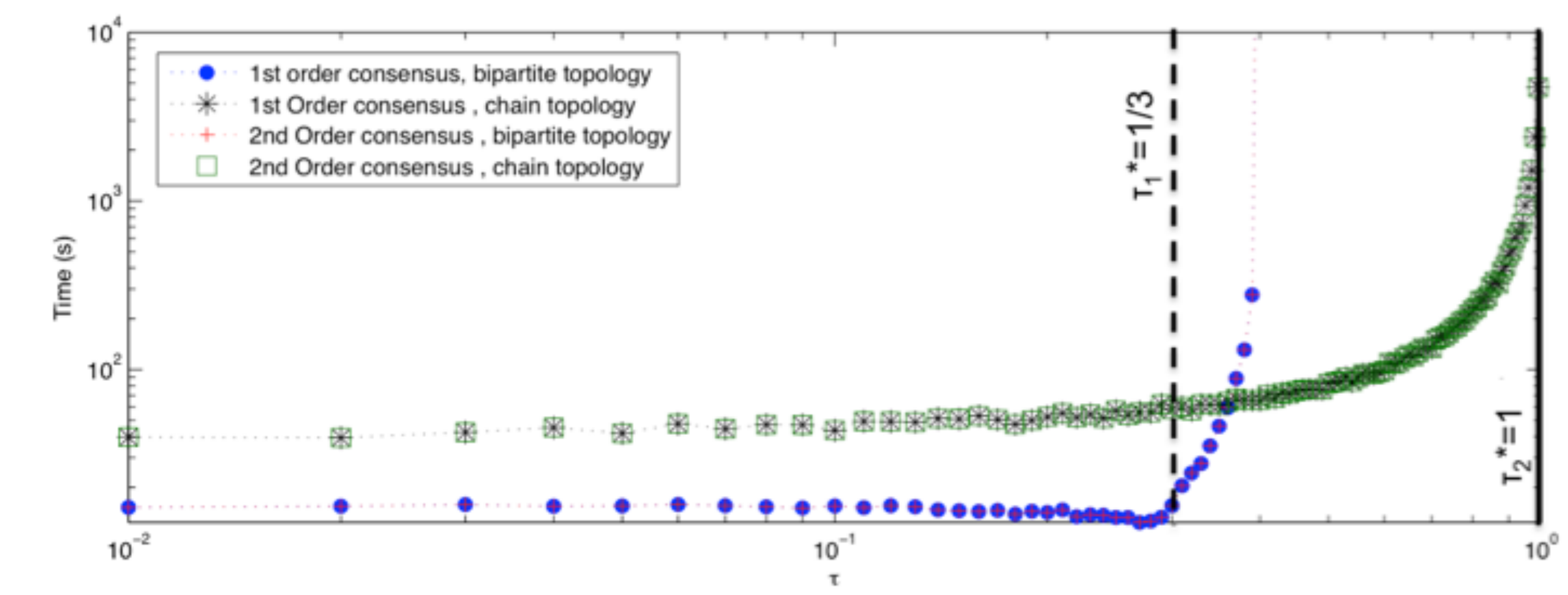}
\caption{Logarithmic plot of time required for consensus depending on the choice of $\tau$ in the interval $(0,1]$. The maximum theoretic limit $\tau_1^*$ for bipartite topology (dotted vertical line) and $\tau_2^*$ for chain topology (solid vertical line)  are highlighted.}
 \label{fig:tau}
\end{center}
\end{figure}

\subsection{Synchronization}

In the field of Critical Infrastructure Protection a crucial aspect is the capability to identify possible risks induced by cascading failures.
Unfortunately, critical infrastructures operators are very reluctant to share detailed information (i.e., from the field data) about their infrastructure, because this data is considered as {\em sensible} (that is to say, the disclosure of such an information would potentially have dramatic consequences on the safety and security of the infrastructures, as well as business and commercial negative effects).

To overcome such a difficulty and provide the operators with a useful tool, in the EU project MICIE \cite{MICIE} we proposed an approach based on a distributed architecture that implements an online risk predictor.
The risk predictor was implemented and tested with respect to a real case study composed of a power grid and a telecommunication infrastructure (see  \cite{MICIE} for a more detailed discussion).

 The main idea of the MICIE system is that the control room of each infrastructure (3 in the proposed case study) is equipped with an identical copy of an abstract and high-level dynamic model able to capture the most relevant domino effects. 
Each model is supplied with data provided by ``its" field (i.e., the data of the infrastructure where the tool is attested), and the different copies have to synchronize exchanging only information about their abstract state.

Specifically, in the proposed case study, we adopted an IIM formulation \cite{Haimes:2001,Haimes:2005a}, as better illustrated in the following.
Each instance of the model, attested in a given infrastructures' control room, acquires as inputs information coming from its own field; in other terms the copy in the control room of the first infrastructure receives as inputs the severity of failure affecting the first infrastructure, the control room of second infrastructure receives as inputs those related to the second infrastructure, and so on.

Allowing to synchronize the different copies, the system is able to provide to the different operators a coherent picture of the global situation, without exchanging sensible information.
Indeed as illustrated in Figure \ref{fig:graph} the different copies do not exchange their complete state; in fact they share a reduced and modified state vector, which depends on the structure of the output matrix $C$. As stated above, while it is desirable to exchange only data  generated within the models, without disclosing field data, this generates nontrivial issues for the choice of the feedback matrix. A more feasible approach is to provide a combination of the two kinds of information; in this way the state exchanged is reduced in dimension and the sensible information is masked.
Finally, due to the level of ambiguity and vagueness that characterizes such model we were forced to consider a DFS formulation. 
\subsection{Input-Output Inoperability Model}
In the literature many interdependency models have been developed, however the Input-Output Inoperability Model (IIM)  \cite{Haimes:2001,Haimes:2005a} gained large attention, because of its simplicity and because of the ability to model cascading effects and indirect dependencies.
In order to provide an indicator of the state of each infrastructure, the \emph{inoperability} $q$ is introduced, as the  inability (in percentage) for an infrastructure to correctly operate.
The IIM model for a scenario composed of $n$ infrastructures is given by:
\begin{eqnarray}
q(k+1)=Aq(k)+ Bc
\label{eq:inoperability_model}
\end{eqnarray} 
where $q,c\in \mathbb{R}^n$; the entires $a_{ij}$ of the $n\times n$ matrix $A$ represent the influence of the inoperability of $j-th$ infrastructure on $i-th$ one. 
The vector $c$ represents external, induced inoperability (it can be seen as a {\it perturbation} generated by an adverse or malicious event), and its effect on the model is mediated by the $n\times n$ matrix $B$.

Notice that, if matrix $A$ is stable, the IIM model reaches an equilibrium \cite{Haimes:2001,Haimes:2005a} $q_{eq}$ given by
 \begin{eqnarray}
q_{eq}=(I_n-A)^{-1}Bc
\label{qeq}
\end{eqnarray} 
such an equilibrium, represents the steady inoperability reached by the infrastructures after considering cascading effects.
\subsection{Simulation Results}
Consider a scenario composed of 3 critical infrastructures, each infrastructure equipped with the same IIM interdependency model, described by the following matrix:
\begin{equation}
A=\begin{bmatrix}
0.1&0.1&0.3\\
0.2&0.1&0.1\\
0.2&0.1&0.2
\end{bmatrix}
\end{equation}
Let us consider the case where the first infrastructure wants to estimate its state and the state of the others based on local information; specifically, let us assume a fuzzy perturbation $c^T=[c_1,0,0]^T$ where $c_1=\{0.05,0.1,0.15\}$ is a triangular fuzzy number and a fuzzy initial condition $q(0)^T=[q_1(0),0,0)]^T$ where $q_1(0)=\{0.05,0.1,0.15\}$ is a triangular fuzzy number.

Note that $\sigma(A)=\{0.4791, -0.1, 0.0209\}$, hence the matrix is stable and has non-negative entries.

Since perturbation $c$ is stationary, it is possible to consider the following extended system, treating $c$ as state variables:

\begin{equation}
w(k+1)=
\begin{bmatrix}
q(k+1)\\
c(k+1)
\end{bmatrix}
=
\begin{bmatrix}
A& B \\ 0 & I_3
\end{bmatrix}
w(k)=\tilde A w(k)
\end{equation}
Note that the above matrix is block triangular, therefore it is stable, since $A$ and $I$ are stable, and has non-negative entires.
Hence conditions required by Theorem \ref{teolin} are satisfied and the fuzzy system characterized by matrix $A$ is stable.

Figure \ref{fig:stability} shows the stable evolution of crisp IIM model (black line) and the stable evolution of the fuzzy IIM model for $\alpha=0$ (stars) and $\alpha=0.5$ (boxes), where crisp initial conditions and perturbations coincide with the central value of triangular fuzzy numbers adopted for $c$ and $q$. Note that, since the infrastructure has only local information, the foreseen inoperability for the other infrastructures is almost zero (i.e., the model only highlights the effect of the inoperability of the first infrastructure on the others).
 \begin{figure}[!ht]
\begin{center}
\includegraphics[width=3.2in]{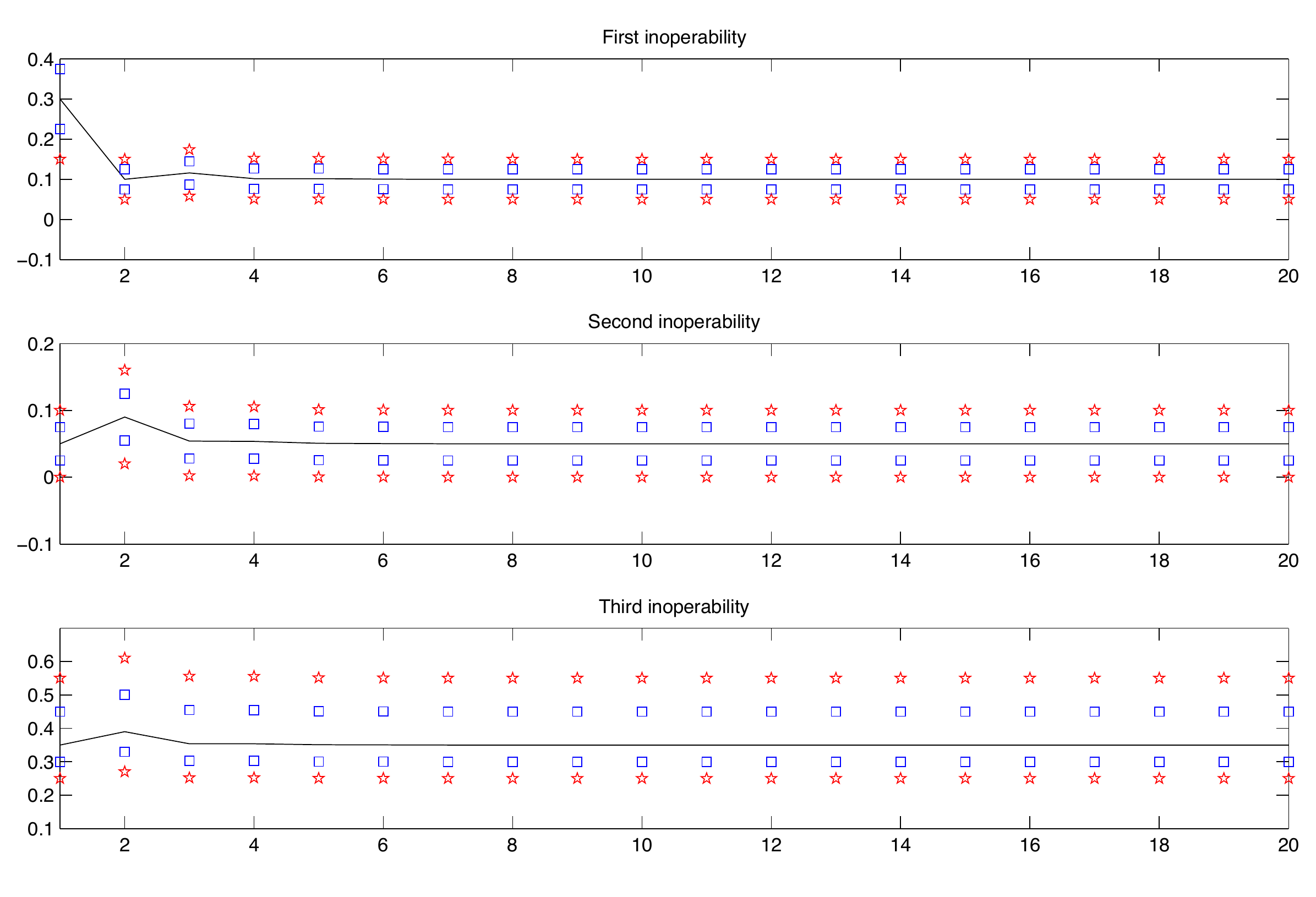}
\caption{Evolution of crisp IIM model (black line) and evolution of the DFS version attested in the first infrastructure for $\alpha=0$ (stars) and $\alpha=0.5$ (boxes). For each $\alpha$-level chosen, the evolution of both left and right extrema of the interval is plotted.}
 \label{fig:stability}
\end{center}
\end{figure}
 
\subsection{Networked systems with total information sharing}
Let us now consider the case where an array of $3$ fuzzy IIM models are interconnected in order to achieve synchronization. To this end let us suppose that each system only receives data coming from its field; in other words the $i$-th system receives an exogenous ``disturbance" $d^i$ that represents the magnitude of the induced perturbation on the $i$-th infrastructure. Therefore each of the $3$ systems has a perturbation $c^i$ where
 
 \begin{equation}
c^1=\begin{bmatrix}
d^1\\ 0 \\ 0
\end{bmatrix}; \quad 
c^2=\begin{bmatrix}
0\\ d^2\\ 0
\end{bmatrix}; \quad 
c^3=\begin{bmatrix}
 0 \\ 0 \\d^3
\end{bmatrix}
\end{equation}
 
Analogously, each of the three systems has its own initial condition $q^i(0)$ where 
 \begin{equation}
q^1(0)=\begin{bmatrix}
q_1^1(0)\\ 0 \\ 0
\end{bmatrix}; \quad 
q^2(0)=\begin{bmatrix}
0\\ q^2_2(0)\\ 0
\end{bmatrix}; \quad 
q^3(0)=\begin{bmatrix}
 0 \\ 0 \\q^3_3(0)
\end{bmatrix}
\end{equation}

In order to achieve synchronization, there is the need to suitably choose the data to be exchanged among systems, i.e., the structure and dimension of matrix $C$; according to Theorem (\ref{teosynchro}), if $(C^TC)$ is non-singular, the choice of the control matrix $\Omega$ is extremely simplified.

Let us first suppose that the systems share their complete information, i.e., $y_i(k)=w_i(k)$ and then $C=I_{2n}$. In this case it is sufficient to set $\Omega=K$ with $K$ diagonal and positive to achieve synchronization.

Let us consider the topology represented  in Figure \ref{fig:graph} (the weights of $\Gamma$ are reported in the figure); as a consequence,  the resulting Laplacian matrix $L$ is given by:
\begin{equation}
\begin{matrix}
L = 
\begin{bmatrix}
1 & -1 & 0\\ -1 & 3 & -2\\0 & -2 &2
\end{bmatrix}
\end{matrix}
\end{equation}

  \begin{figure}[!ht]
\begin{center}
\includegraphics[width=3.2in]{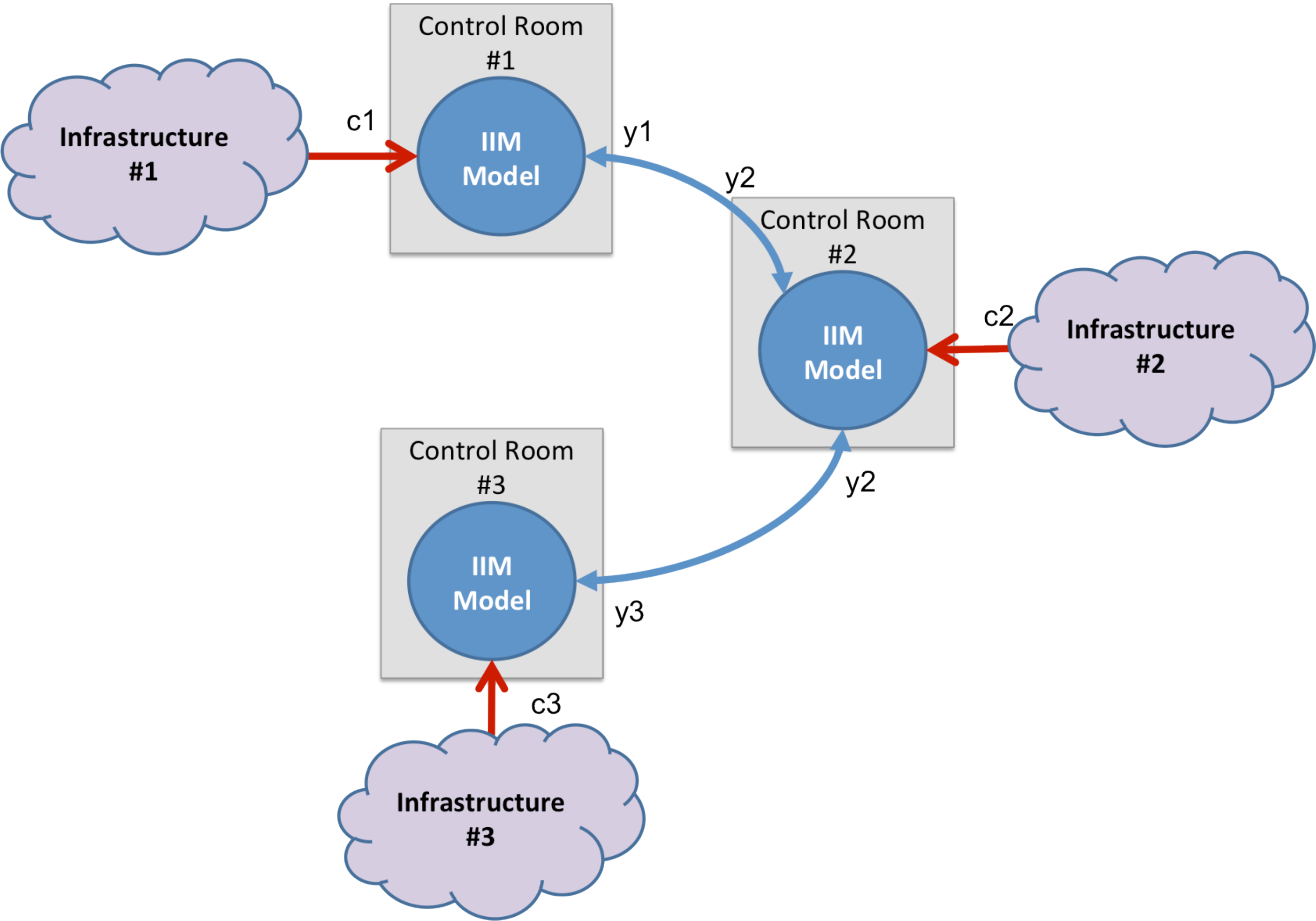}
\caption{In the proposed case study three systems, each equipped with an IIM model, are interconnected by means of weighted edges. Moreover, each system only receives inputs originated within its field.}
 \label{fig:graph}
\end{center}
\end{figure}

Note that $l_m=1$ (i.e., the minimum diagonal element of $L$); then since the minimum diagonal entry of matrix $A$ is $0.1$, according to theorem (\ref{syncfz}) the entries of the diagonal matrix $K$ should be $k_{ii}\leq 0.1$, in order to grant synchronization.
Therefore in this case $\Omega=K$, with $K$ diagonal and such that its diagonal entries $k_{ii}=0.03$ for all $i$.

As stated before, unfortunately, a complete information sharing approach, although effective, represents an unfeasible solution, since the operators are reluctant to completely disclose the infrastructure state. 
\subsection{Networked systems with total information sharing}

To overcome the above drawback, there is the need to provide an output matrix $C$ of reduced dimensions, in order to disclose as little information as possible.
To this end, let us introduce the \emph{expected inoperability ratio} ($EIR$) at time step $k$  for $i$-th system as follows:
 \begin{eqnarray}
EIR_i(k)=(I-A)^{-1}Bc_i(k)-q_i(k)
\label{eir}
\end{eqnarray} 
Recall from Eq. (\ref{qeq}) that the equilibrium reached by the isolated IIM system is given by $(I-A)^{-1}Bc$.
Hence (\ref{eir}) represents the difference between the steady state inoperability foreseen for the $i$-th copy and the actual degree of inoperability estimated by such a IIM copy, i.e., $q_i(k)$.

Due to its characteristics, this index provides only very general information, without disclosing the actual data.
Notice that the case $EIR_{ij}(k)=0$ represents two completely different situations; the case in which the $j$-th infrastructure is completely working, and the case in which the expected equilibrium coincides with the actual inoperability. This ambiguity emphasizes that the $EIR$ index alone does not guarantee the observability of the system. Indeed the corresponding $n\times 2n$ matrix $C$ is in the form $C=\begin{bmatrix}-I_n & (I-A)^{-1}B\end{bmatrix}$, and does not satisfy the condition on $C^TC$.
There is therefore the need to consider also at least another variable to be exchanged among systems; still, in order to avoid disclosure of sensible data, we chose to let each system $i$ exchange the average of its inoperability vector $q_i$, i.e., $\hat q_i(k)=\frac{1}{n}\sum_{j=1}^n q_{ij}(k)$. 
With this choice the resulting $(n+1)\times 2n$ matrix $C$ becomes:
\begin{equation}
\label{pppo}
C=\begin{bmatrix}-I_n & (I-A)^{-1}B\\ \frac{1}{n}\cdots \frac{1}{n}&0 \cdots 0\end{bmatrix}
\end{equation}
Let $H=(I-A)^{-1}B$; $C^TC$ is given by
\begin{equation}
\label{pppo}
C^TC=\begin{bmatrix}(1+\frac{1}{n^2})I_n & -H\\ -H^T&H^TH\end{bmatrix}
\end{equation}
It is a standard result that if $P, Q, R, S$ are square matrix of the same dimensions and $P$ is invertible then
$$
det(\begin{bmatrix}P & Q \\ R & S\end{bmatrix})=det(P)det(S-RP^{-1}Q)
$$
therefore
$$det(C^TC)=det( (1+\frac{1}{n^2})I_n)det(\frac{1}{n^2}H^TH)=2\frac{n^2+1}{n^4}det(H)$$
which is nonzero, since det(H) is non zero; therefore $C^TC$ is nonsingular.

The particular $C$ matrix chosen therefore, allows synchronization without disclosing sensible information, since only aggregate data is shared.
Note that, setting $\Omega=KC^\dag$, we have that $\Omega C= K$, therefore using exactly the same matrix $K$ defined above, the $3$ systems have exactly the same evolution with respect to the total information case.

Let us now provide a simulation example.
In our simulation the following values were considered for the disturbances:
\begin{equation}
\begin{matrix}
d^1&=&\{0.05,0.1,0.15\}\\ d^2&=&\{0,0.05,0.1\}\\ d^3&=&\{0.25,0.35,0.55\} 
\end{matrix}
\end{equation}
moreover, we chose to set $q_1^1(0)=\{0.1,0.2,0.3\}$, while the other systems had zero initial conditions. 

As stated above, in order to reconstruct the evolution of the whole system, there is the need to adopt for each system a modified initial condition $\hat w_i(0)= 3 w_i(0)$.
Note that the width of the resulting fuzzy number is greater than the initial one; this means that uncertainty increases during the evolution of systems, just as expected.

In Figure \ref{fig:levelwise} the synchronization of left and right extrema of the $\alpha$-levels of $q_1$ and $c_1$ are plotted for each system in the case $\alpha=0.9$ and $\alpha=0.5$.

\begin{figure}[!ht]
\begin{center}
\includegraphics[width=3.2in]{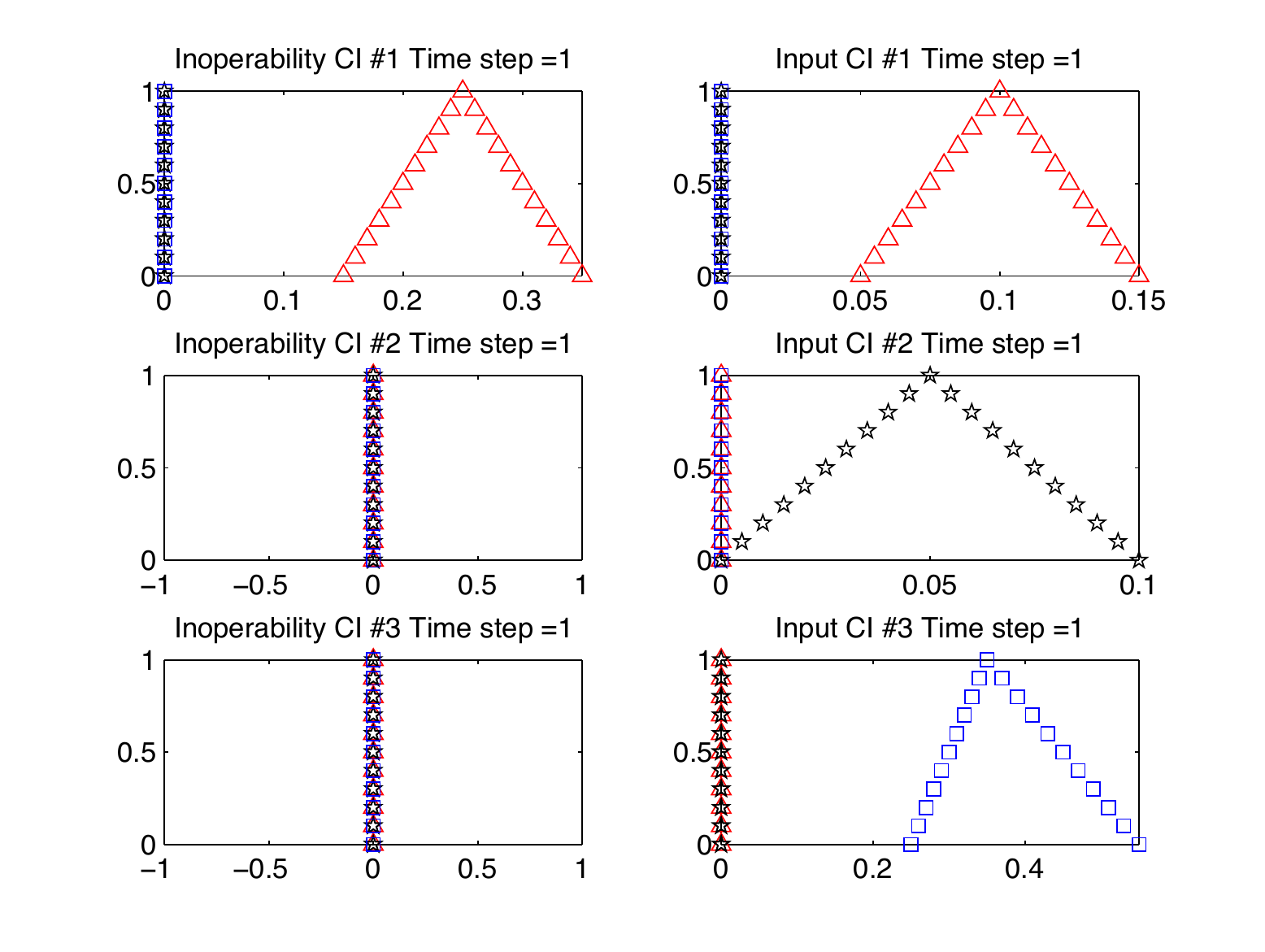}
\caption{Initial state for each state variable and each system (triangles, stars and boxes represent system $1$, $2$ and $3$, respectively).}
 \label{fig:allstart}
\end{center}
\end{figure}

\begin{figure}[!ht]
\begin{center}
\includegraphics[width=3.27in]{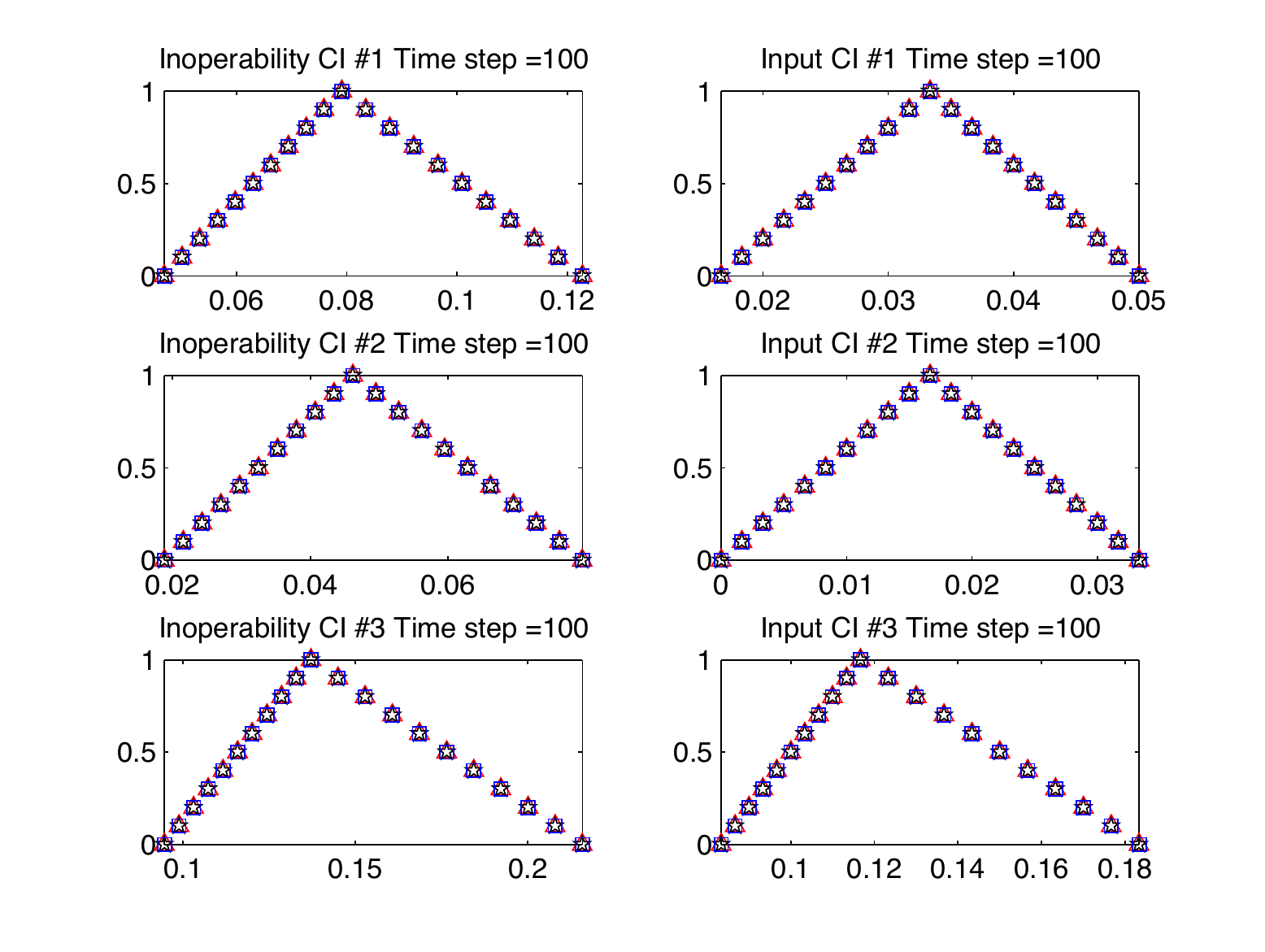}
\caption{Final synchronized state for each state variable and system.}
 \label{fig:allstop}
\end{center}
\end{figure}

\begin{figure}[!ht]
\begin{center}
\includegraphics[width=3.2in]{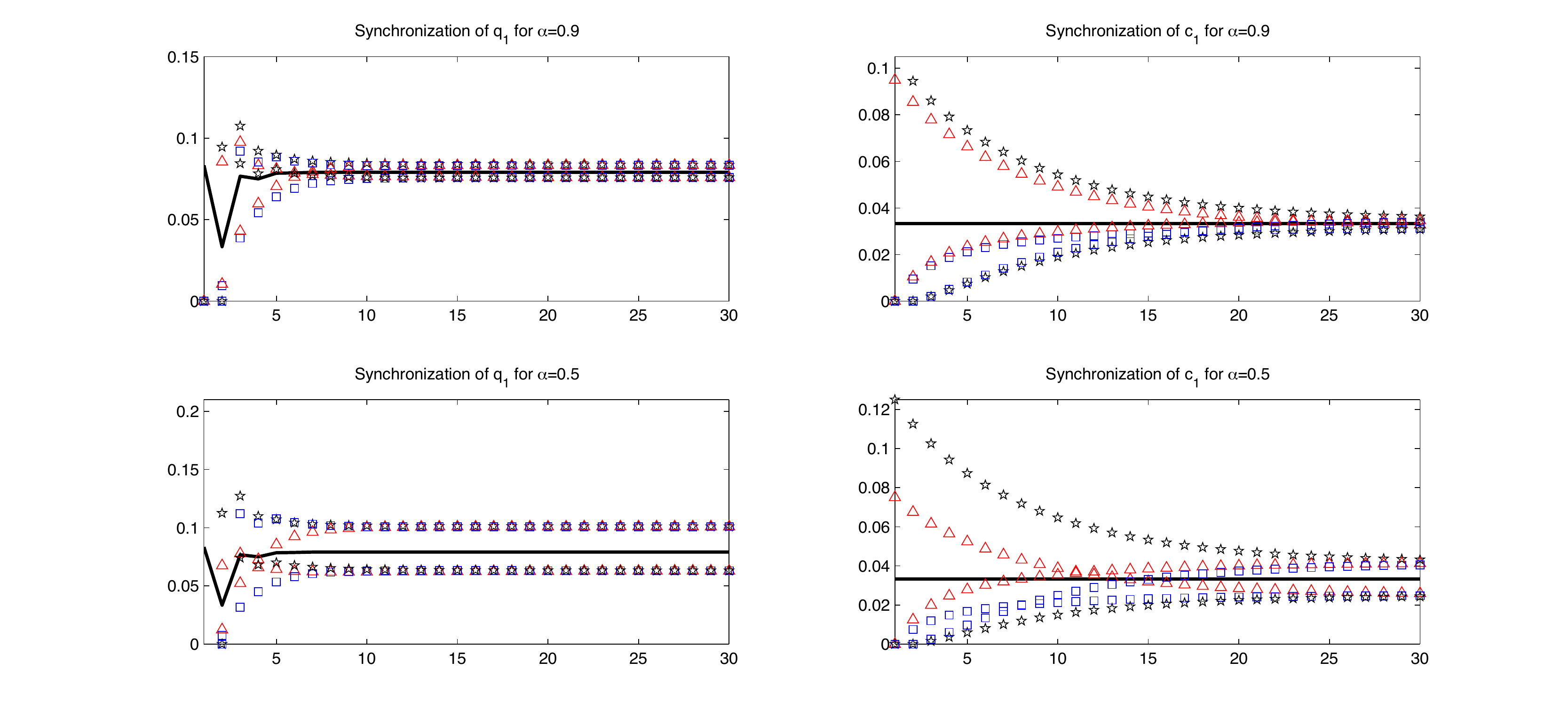}
\caption{ Synchronization of $q_1$ and $c_1$ for the three systems (triangles, stars and boxes represent system $1$, $2$ and $3$, respectively) with $\alpha=0.9$ and $\alpha=0.5$. The evolution of the crisp system with the composition of the initial conditions and the composition of the perturbations (the central values of the fuzzy numbers) is also plotted with bold line.}
 \label{fig:levelwise}
\end{center}
\end{figure}

To conclude this Section, note that the distributed approach is indeed able to capture the evolution of all the infrastructures (e.g.,  in the distributed approach the third infrastructure has a steady fuzzy inoperability whose peak value is $0.55$, while in the isolated example of Section 6.2 the peak value of the third infrastructures' inoperability was about $0.025$). This is particularly true considering the first infrastructure; in fact the peak value of the steady inoperability reached in the distributed case is $0.31$ while in the isolated case is about $0.08$, since in the distributed case the actual domino effects are taken into account. 

\section{Conclusions}
\label{conclusions}

In this paper the distributed synchronization and consensus problems have been extended in the fuzzy fashion, in order to manage uncertainty and vague information, with the aim to provide a distributed tool for the analysis of the state of critical infrastructures, when human operators and actors are directly involved, by means of partial state observations, as well as a framework for the synchronization of arrays of fuzzy interdependency models.

Future work will be devoted to extend the framework, in order to take into account distributed systems with uncertain dynamics.

\appendix
\section{Proofs}
\label{proofs}

\noindent {\bf Proof of Theorem \ref{lem:tfn}}

\noindent It is sufficient to show that the set of triangular fuzzy numbers is a subset of $\Psi$.
Let a triangular fuzzy number $\mu_t=\{\mu_l, \mu_c, \mu_r\}$ and consider $p,q \in \mathbb{R}, p\leq q$; let $h=\phi p + (1-\phi)q, \phi \in [0,1]$.
First of all consider the case in which $p\leq q \leq \mu_c$ or $\mu_c \leq p \leq q$. The value of $\mu_t(h)$ is given by:
\begin{equation}
\label{eq:rect}
\mu_t(h)=\frac{\mu_t(q)-\mu_t(p)}{q-p}(h-p)+\mu_t(p)
\end{equation}
substituting $h$ with $\phi p + (1-\phi)q$, condition (\ref{eq:psieq}) is satisfied as an equivalence.
Consider the case in which $p\leq \mu_c \leq q$ and $\mu_t(p)\leq \mu_t(q)$. In this case there exists a $q^*$ such that $p \leq h \leq q^*$ implies condition (\ref{eq:psieq}) is satisfied as an equivalence, while the case $q^* < h \leq q$ implies condition (\ref{eq:psieq}) is satisfied.
A similar result holds in the case in which $p\leq \mu_c \leq q$ and $\mu_t(p) > \mu_t(q)$, proving the statement.

\noindent {\bf Proof of Theorem \ref{lem:stab3}}

\noindent First, there is the need to prove that, under the hypotheses
\begin{equation}
\label{implpippo} 
V(x(0),0)\leq z(0) \Rightarrow V(x(k+1),k+1)\leq z(k+1); \quad  \forall k\geq 0
\end{equation}
In \cite{Laksh:2003}, Theorem 5.2.1, it is proven that, given a scalar function $g(r,k)$ nondecreasing in $r$ for each $k\geq 0, k\in \mathbb{N}^+$ and given two sequences of real numbers $\{c_k\}, \{d_k\}$ such that $c_0\leq d_0$ and such that the following inequality holds for all $k\geq 0$
\begin{eqnarray}
\label{xxx}
c_{k+1}\leq g(c_k,k)\\
\label{yyy}
d_{k+1}\geq g(d_k,k)
\end{eqnarray}
then $c_k\leq d_k$, for all $k\geq 0$. 

Such a Theorem trivially extends to the vectorial case, therefore, considering a vectorial $G(r,k):\mathbb{R}^N\times \mathbb{N}^+\rightarrow \mathbb{R}^N$ and setting $\{c_k\}$ as the sequence of defuzzyfied values $\{V(x(k),k)\}$ and $\{d_k\}$ equal to the sequence of values assumed by system (\ref{fuzzysystem}) $\{z(k)\}$, implication (\ref{xxx}) holds becuse of implication (\ref{dpiuineq}) and implication (\ref{yyy}) is true due to the monotonicity of $G(\cdot,\cdot)$; therefore implication (\ref{implpippo}) is proved.

Suppose that the trivial solution of (\ref{eq:crispdiffsist}) is stable. Then, for each $a(\epsilon)>0$ there exists a positive $\delta_1(\epsilon)$ such that
\begin{equation}
\label{fimpl}
\begin{matrix}
d_{\mathbb{R}^N}[z(0),0_N]=\sum_{j=1}^N z_j(0)<\delta_1(\epsilon) \Rightarrow \\
 \Rightarrow d_{\mathbb{R}^N}[z(k+1),0_N]=\sum_{j=1}^N z_j(k+1)<a(\epsilon)
\end{matrix}
\end{equation}
where $d_{\mathbb{R}^N}[\cdot,\cdot]$ is the distance in $\mathbb{R}^N$.

To prove the stability of Sytem (\ref{fuzzysystem}) there is the need to show that, for any $\epsilon\geq 0$ there exists a  positive $\delta(\epsilon)$ such that if $d_{\mathbb{E}^N}[x(0),\hat 0]\leq \delta(\epsilon)$ then $d_{\mathbb{E}^N}[x(k),\hat 0]<\epsilon$, for each $k\geq 0$.

Since $z(0)= V(x(0),0)$, on the base of the Implication (\ref{implpippo}), it follows that $V_0(x(k+1),k+1)\leq \sum_{j=1}^Nz(k+1)$; therefore, according to Inequality (\ref{bound_condition}):

\begin{eqnarray}
\begin{matrix}
a(d_{\mathbb{E}^N}[x(k+1),\hat 0])\leq V_0(x(k+1),k+1)\leq \\
\leq \sum_{j=1}^N z_j(k+1)< a(\epsilon)
\end{matrix}
\end{eqnarray}
Due to the continuity and monotonicity of $a(\cdot)$, it follows that 
\begin{eqnarray}
a(d_{\mathbb{E}^N}[x(k+1),\hat 0])< a(\epsilon) \Rightarrow d_{\mathbb{E}^N}[x(k+1),\hat 0] < \epsilon
\end{eqnarray}
and the stability of system (\ref{fuzzysystem}) is proved.

For asymptotic stability note that 
\begin{eqnarray*}
0\leq a(d_{\mathbb{E}^N}[x(k+1),\hat 0])\leq V_0(x(k+1),k+1)\leq \\
\leq \sum_{j=1}^N z_j(k+1)
\end{eqnarray*}
If System (\ref{eq:crispdiffsist}) is asymptotically stable, then $z_j(k+1)\rightarrow 0$ as $k\rightarrow \infty$, for each $j=1,\ldots,m$, and it follows that $d_{\mathbb{E}^N}[x(k+1),\hat 0]\rightarrow 0$ as $k\rightarrow \infty$, too. 

\noindent {\bf Proof of Corollary \ref{teolin}}

\noindent It will be shown that, choosing $V_i({ x}(k),k)=d_{\mathbb{E}}(x_i(k),\hat 0_i)$ the conditions required by Theorem \ref{lem:stab3} are verified.

Inequality (\ref{bound_condition}) is satisfied if it is true component-wise, that is, if:
\begin{equation}
a(d_{\mathbb{E}^n}({ x}(k),\hat { 0}))\leq \sum_{i=1}^n d_{\mathbb{E}}(x_i(k),\hat0_i) =d_{\mathbb{E}^n}({ x}(k),\hat { 0})
\end{equation}
Since $d_{\mathbb{E}^n}({ x}(k),\hat { 0})\geq 0$, the inequality is verified choosing $a(r)=\psi r$, with $\psi \in [0,1]$. In this case $a(\cdot)$ is continuous and monotone non-decreasing.
It remains to prove inequality (\ref{dpiuineq}).

Note that $$V(x(k),k)=[d_{\mathbb{E}}(x_1(k),\hat 0_1), \ldots, d_{\mathbb{E}}(x_n(k),\hat 0_n)]^T$$ There is the need to prove that, for all $i=1, \ldots, n$
\begin{equation}
\label{TBPved}
d_{\mathbb{E}}(f_i({ x}(k),k),\hat 0_i) \leq f_i(V({ x}(k),k),k)
\end{equation}
Note that $[\hat 0_i]^\alpha=\{0\}$ for all $\alpha$, hence:
\begin{equation}
\begin{matrix}
{ d}_{\mathbb{E}}(w,\hat { 0}_i)=\sup_{\alpha >0}\{\rho_{d_\mathbb{R},\mathbb{R}}([w]^\alpha,\{0\}) \}=\\ \\
=\max\{d_{\mathbb{R}}([w]^0,\{0\})\}=\max\{||z|| : z \in [w]^0\}
\end{matrix}
\end{equation}
that is to say that the distance of a given fuzzy set is the maximum value of the norm of its support (i.e., the absolute value of one of the endpoints of the support).
Analogously we have that
\begin{equation}
V_i({ x}(k),k)=\max\{||z|| : z \in [x_i(k)]^0\}
\end{equation}

Note that ${ f}_i(\cdot,k)$ is monotone nondecreasing and that in the left term of inequality \eqref{TBPved} the $\max \{||\cdot||\}$ are taken component-wise; therefore inequality \eqref{TBPved} is always verified.

Hence the requirements of Theorem \ref{lem:stab3} are fulfilled and the proof is complete.

\bibliography{biblio}{}

\begin{thebibliography}{10}

\bibitem{Synchro4}
I.~Belykh, V.~Belykh, and M.~Hasler.
\newblock ''generalized connection graph method for synchronization in
  asymmetrical networks''.
\newblock {\em Physica D}, pages 42--51, 2006.

\bibitem{WeiRen2}
Y.~Cao and W.~Ren.
\newblock Sampled-data discrete-time coordination algorithms for
  double-integrator dynamics under dynamic directed interaction.
\newblock {\em Int. Journal of Control}, 83(3):506--515, 2009.

\bibitem{Haimes:2005a}
Y.~Haimes, B.~Horowitz, J.~Lambert, J.~Santos, C.~Lian, and K.~Crowther.
\newblock Inoperability input-output model for interdependent infrastructure
  sectors i: Theory and methodology.
\newblock {\em Journal of Infrastructure Systems}, 11(2):67--79, 2005.

\bibitem{Haimes:2001}
Y.~Haimes and P.~Jiang.
\newblock Leontief-based model of risk in complex interconnected
  infrastructures.
\newblock {\em Journal of Infrastructure Systems}, 7(1):1--12, 2001.

\bibitem{Kay}
H.~Kay and B.~Kuipers.
\newblock Numerical behavior envelopes for qualitative models.
\newblock In {\em Proceedings of 11th National Conference on Artificial
  Intelligence}, pages 606--613, 1993.

\bibitem{Laksh:2003}
V.~Lakshmikantham and R.~N. Mohapatra.
\newblock {\em Theory of fuzzy differential equations and inclusions}.
\newblock CRC Press LLC, 2003.

\bibitem{Trig}
V.~Lakshmikantham and D.~Trigiante.
\newblock {\em Theory of Difference Equations: Numerical Methods and
  Applications}.
\newblock Academic Press, New York, 1988.

\bibitem{Lewis}
T.~G. Lewis.
\newblock {\em Critical Infrastructure Protection in Homeland Security:
  Defending a Networked Nation}.
\newblock Wiley, 2006.

\bibitem{Olfati12}
R.~Olfati-Saber.
\newblock Flocking for multi-agent dynamic systems: algorithms and theory.
\newblock {\em IEEE Trans. Autom. Control}, 51(3):401--420, 2006.

\bibitem{Olfati1}
R.~Olfati-Saber and R.~M. Murray.
\newblock Consensus problems in networks of agents with switching topology and
  time-delays.
\newblock {\em IEEE Transactions on Automatic Control}, 49(9):1520--1533, 2004.

\bibitem{Pearson}
D.~Pearson.
\newblock A property of linear fuzzy differential equations.
\newblock {\em Appl. Math. Lett}, pages 99--103, 1997.

\bibitem{Synchro2}
L.~Pecora and T.~Carroll.
\newblock Master stability functions for synchronized coupled systems.
\newblock {\em Phys. Rev. Lett}, pages 2109--2112, 1998.

\bibitem{MICIE}
MICIE Project.
\newblock website.
\newblock http://www.micie.eu.

\bibitem{Olfati2}
R-Olfati-Saber, J.~A. Fax, and R.~M. Murray.
\newblock Consensus and cooperation in networked multi-agent systems.
\newblock {\em Proceedings of the IEEE}, 95(1):215--233, 2007.

\bibitem{SetolaMarino}
S.~De~Porcellinis R.~Setola and M.~Sforna.
\newblock Critical infrastructure dependency assessment using input-output
  inoperability mode.
\newblock {\em Int. J. Critical Infrastructure Protection (IJCIP)}, pages 170
  -- 178, 2009.

\bibitem{WeiRen01}
W.~Ren.
\newblock Consensus strategies for cooperative control of vehicle formations.
\newblock {\em IET Control Theory \& Applications}, 1(2):505--512, March 2007.

\bibitem{WeiRen0}
W.~Ren.
\newblock On consensus algorithms for double-integrator dynamics.
\newblock {\em IEEE Trans. Autom. Control}, 53(6):15603--1509, 2008.

\bibitem{WeiRen}
W.~Ren and R.W. Beard.
\newblock {\em Possibility Theory: an Approach to Computerized Processing of
  Uncertainty}.
\newblock Plenum Publishing Corporation, 1998.

\bibitem{WeiRen1}
W.~Ren and Y.~Cao.
\newblock Convergence of sampled-data consensus algorithms for
  double-integrator dynamics.
\newblock {\em Proc. of 47th IEEE Conf. on Decision and Control}, pages 9--11,
  2008.

\bibitem{Kelly:2001}
S.~Rinaldi, J.~Peerenboom, and T.~Kelly.
\newblock Identifying understanding and analyzing critical infrastructure
  interdependencies.
\newblock {\em IEEE Control System Magazine}, 21(6):11--25, 2001.

\bibitem{Seikkala}
S.~Seikkala.
\newblock On the fuzzy initial value problem.
\newblock {\em Fuzzy Sets and Systems}, 24:319--330, 1987.

\bibitem{Tuna2}
S.~E. Tuna.
\newblock Synchronizing linear systems via partial-state coupling.
\newblock {\em Automatica}, 44(8):2179--2184, 2008.

\bibitem{Tuna1}
S.~E. Tuna.
\newblock Conditions for synchronizability in arrays of coupled linear systems.
\newblock {\em IEEE Transactions on Automatic Control}, pages 2416--2420,
  October 2009.

\bibitem{CIRCLE}
R.~S. Varga.
\newblock {\em Gershgorin and his circles}.
\newblock Springer-Verlag, 2004.

\bibitem{Wu:2005}
C.~Wu.
\newblock Synchronization in networks of nonlinear dynamical systems coupled
  via a directed graph''.
\newblock {\em Nonlinearity}, pages 1057--1064, 2005.

\end{thebibliography}
\bibliographystyle{plain}

\end{document}